\documentclass[a4paper,UKenglish,cleveref, autoref, thm-restate]{lipics-v2021}

\title{Dominating Set Knapsack: Profit Optimization on Dominating Sets} 
\author{Sipra Singh}{Indian Institute of Technology Kharagpur, India }{sipra.singh@iitkgp.ac.in}{}{}

\authorrunning{S. Singh} 
\Copyright{} 



\category{} 

\relatedversion{} 

\EventEditors{John Q. Open and Joan R. Access}
\EventNoEds{2}
\EventLongTitle{42nd Conference on Very Important Topics (CVIT 2016)}
\EventShortTitle{CVIT 2016}
\EventAcronym{CVIT}
\EventYear{2016}
\EventDate{December 24--27, 2016}
\EventLocation{Little Whinging, United Kingdom}
\EventLogo{}
\SeriesVolume{42}
\ArticleNo{23}
\usepackage{hyperref}
\usepackage[usenames,svgnames,dvipsnames]{xcolor}

\hypersetup{
	colorlinks=true,
	linkcolor=magenta,
	filecolor=blue,      
	urlcolor=cyan,
	citecolor=Emerald
	}
\usepackage{amssymb,amsmath,amsthm}
\usepackage{multirow}
\usepackage{eulervm}
\usepackage{graphicx}
\usepackage{bbm}
\usepackage{setspace}

\usepackage{tikz}
\usetikzlibrary{arrows,positioning}

\newdimen\prevdp
\def\leftlabel#1{\noalign{\prevdp=\prevdepth
   \kern-\prevdp\nointerlineskip\vbox to0pt{\vss\hbox{\ensuremath{#1}}}\kern\prevdp}}

\usepackage{url}

\usepackage{enumerate}

\usepackage{xspace}
\newcommand{\Pc}{\ensuremath{\mathsf{P}}\xspace}
\newcommand{\NP}{\ensuremath{\mathsf{NP}}\xspace}
\newcommand{\NPC}{\ensuremath{\mathsf{NP}}\text{-complete}\xspace}

\newcommand{\wh}{{\ensuremath{\mathsf{W}}\sc{-hierarchy}}\xspace}
\newcommand{\WO}{\ensuremath{\mathsf{W[1]}}\xspace}
\newcommand{\WTw}{\ensuremath{\mathsf{W[2]}}\xspace}
\newcommand{\WTh}{\ensuremath{\mathsf{W[3]}}\xspace}

\newcommand{\FPT}{\ensuremath{\mathsf{FPT}}\xspace}
\newcommand{\ETH}{\ensuremath{\mathsf{ETH}}\xspace}
\newcommand{\ExpoTH}{{\sc Exponential Time Hypothesis}\xspace}
\newcommand{\SETH}{\ensuremath{\mathsf{SETH}}\xspace}
\newcommand{\wcs}{{\sc Weighted Circuit Satisfiability}\xspace}
\newcommand{\WCS}{\ensuremath{\mathsf{WCS}}\xspace}

\newcommand{\PTAS}{\ensuremath{\mathsf{PTAS}}\xspace}

\let\oldlambda\lambda
\renewcommand{\lambda}{\ensuremath{\oldlambda}\xspace}
\let\oldalpha\alpha
\renewcommand{\alpha}{\ensuremath{\oldalpha}\xspace}
\let\oldDelta\Delta
\renewcommand{\Delta}{\ensuremath{\oldDelta}\xspace}

\newcommand{\YES}{{\sc yes}\xspace}
\newcommand{\NO}{{\sc no}\xspace}
\newcommand{\yes}{{\sc yes}\xspace}
\newcommand{\no}{{\sc no}\xspace}
\newcommand{\true}{\text{{\sc true}}\xspace}

\newcommand{\kp}{\text{\sc Knapsack Problem}\xspace}

\newcommand{\indset}{\text{\sc Independent Set}\xspace}

\newcommand{\vc}{\text{\sc Vertex Cover}\xspace}

\newcommand{\vcknapsack}{{\sc Vertex Cover Knapsack}\xspace}

\newcommand{\dos}{\text{\sc Dominating Set}\xspace}
\newcommand{\udos}{{\sc Upper Dominating Set}\xspace}
\newcommand{\dosknapsack}{{\sc Dominating Set Knapsack}\xspace}
\newcommand{\kdosknapsack}{{\sc \ensuremath{k}-Dominating Set Knapsack}\xspace}
\newcommand{\minimaldosknapsack}{{\sc Minimal Dominating Set Knapsack}\xspace}

\newcommand{\td}{{\sc Tree decomposition}\xspace}
\newcommand{\ntd}{{\sc Nice tree decomposition}\xspace}
\newcommand{\tw}{{\sc Treewidth}\xspace}

\newcommand{\BB}{\ensuremath{\mathcal B}\xspace}
\newcommand{\CC}{\ensuremath{\mathcal C}\xspace}
\newcommand{\DD}{\ensuremath{\mathcal D}\xspace}
\newcommand{\EE}{\ensuremath{\mathcal E}\xspace}
\newcommand{\FF}{\ensuremath{\mathcal F}\xspace}
\newcommand{\GG}{\ensuremath{\mathcal G}\xspace}
\newcommand{\HH}{\ensuremath{\mathcal H}\xspace}
\newcommand{\II}{\ensuremath{\mathcal I}\xspace}

\newcommand{\OO}{\ensuremath{\mathcal O}\xspace}

\renewcommand{\SS}{\ensuremath{\mathcal S}\xspace}
\newcommand{\TT}{\ensuremath{\mathcal T}\xspace}
\newcommand{\UU}{\ensuremath{\mathcal U}\xspace}
\newcommand{\VV}{\ensuremath{\mathcal V}\xspace}
\newcommand{\WW}{\ensuremath{\mathcal W}\xspace}
\newcommand{\XX}{\ensuremath{\mathcal X}\xspace}

\usepackage{nicefrac}

\newcommand{\eps}{\ensuremath{\varepsilon}\xspace}
\renewcommand{\epsilon}{\eps}

\newcommand{\ignore}[1]{}

\newcommand{\pr}{\ensuremath{\prime}}

\renewcommand{\leq}{\leqslant}
\renewcommand{\geq}{\geqslant}
\renewcommand{\ge}{\geqslant}
\renewcommand{\le}{\leqslant}

\usepackage{cleveref}

\crefname{theorem}{Theorem}{\bf Theorems}
\crefname{observation}{Observation}{\bf Observations}
\crefname{lemma}{Lemma}{\bf Lemmata}
\crefname{corollary}{Corollary}{\bf Corollaries}
\crefname{proposition}{Proposition}{\bf Propositions}
\crefname{definition}{Definition}{\bf Definitions}
\crefname{claim}{Claim}{\bf Claims}
\crefname{reductionrule}{Reduction rule}{\bf Reduction rules}

\usepackage{mathdots}
\usepackage{thesis}
\usepackage{tabularx}

\usepackage{comment}

\usepackage{rotating}
\usepackage{amssymb}
\usepackage{latexsym}
\usepackage{epsfig}

\usepackage{url}
\usepackage{caption}
\usepackage{subcaption}
\usepackage{array}
\usepackage{multirow}
\usepackage{enumerate}
\usepackage{pdflscape}
\usepackage{amsmath,graphicx,algorithm,algpseudocode,amsfonts}
\usepackage{epstopdf}

\usepackage{float}
\floatstyle{ruled}
\newfloat{Algorithms}{!thb}{lop}
\floatname{Algorithms}{Algorithm}

\allowdisplaybreaks[1]

\algnewcommand\algorithmicinput{\textbf{Input:}}
\algnewcommand\INPUT{\item[\algorithmicinput]}
\algnewcommand\algorithmicoutput{\textbf{Output:}}
\algnewcommand\OUTPUT{\item[\algorithmicoutput]}
\algnewcommand{\LineComment}[1]{\State \(\triangleright\) #1}

\usepackage{setspace}
\usepackage{amsthm}
\usepackage{pifont}

\usepackage{mathtools}

\usepackage{algorithm}
\usepackage{algpseudocode}
\usepackage{pifont}
\usepackage{graphicx}
\graphicspath{ {./pdf/} }
\usepackage[utf8]{inputenc}
\usepackage{libertine}

\makeatletter
\newenvironment{breakablealgorithm}
{
		\begin{center}
			
			\refstepcounter{algorithm}
			\hrule height.8pt depth0pt \kern2pt
			\renewcommand{\caption}[2][\relax]{
				{\raggedright\textbf{\fname@algorithm~\thealgorithm} ##2\par}%
				\ifx\relax##1\relax 
				\addcontentsline{loa}{algorithm}{\protect\numberline{\thealgorithm}##2}%
				\else 
				\addcontentsline{loa}{algorithm}{\protect\numberline{\thealgorithm}##1}%
				\fi
				\kern2pt\hrule\kern2pt
			}
		}{
		\kern2pt\hrule\relax
	\end{center}
}
\makeatother

\usepackage[normalem]{ulem}

\usepackage{makecell}
\usepackage{cleveref}
\crefname{theorem}{Theorem}{\bf Theorems}
\crefname{observation}{Observation}{\bf Observations}
\crefname{lemma}{Lemma}{\bf Lemmata}
\crefname{corollary}{Corollary}{\bf Corollaries}
\crefname{proposition}{Proposition}{\bf Propositions}
\crefname{definition}{Definition}{\bf Definitions}
\crefname{claim}{Claim}{\bf Claims}
\crefname{reductionrule}{Reduction rule}{\bf Reduction rules}
\usepackage{color,soul} 
\usepackage{bbm}
\usepackage{relsize} 
\usepackage{lipsum}

\usepackage{etoolbox}
\nolinenumbers
\begin{document}
	
	\maketitle
	
\begin{abstract}
In a large-scale network, we want to choose some influential nodes to make a profit by paying some cost within a limited budget so that we do not have to spend more budget on some nodes adjacent to the chosen nodes; our problem is the graph-theoretic representation of it. We define our problem, \dosknapsack, by attaching the \kp with the \dos on graphs. Each vertex $v~(\in \VV) $ is associated with a cost factor $w(v)$ and a profit amount $\alpha(v)$. We aim to choose some vertices within a fixed budget $(s)$ that give maximum profit so that we do not need to choose their 1-hop neighbors. We show that the \dosknapsack problem is strongly \NPC even when restricted to bipartite graphs, but weakly \NPC for star graphs. We present a pseudo-polynomial time algorithm for trees in time $\OO(n\cdot \min\{s^2, (\alpha(\VV))^2\})$. We show that \dosknapsack is unlikely to be \textit{fixed parameter tractable} by proving that it is \WTw-hard parameterized by the solution size. We developed \FPT algorithms with running time $\OO(4^{tw}\cdot n^{\OO(1)} \min\{s^2,{\alpha(\VV)}^2\})$ and $\OO(2^{vck-1}\cdot n^{\OO(1)} \min\{s^2,{\alpha(\VV)}^2\})$, where $tw$ represents the \tw of the given graph $\GG(\VV,\EE)$, $vck$ is the solution size of the \vcknapsack, $s$ is the capacity or size of the knapsack and $\alpha(\VV)=\sum_{v\in\VV}\alpha(v)$. We obtained similar results for other variants \kdosknapsack and \minimaldosknapsack, where $k$ is the size of the \dos.

\end{abstract}
%
%
%


\section{Introduction}\label{sec:intro}
The \kp is a fundamental problem widely studied in Computer Science. In this problem, there are $n$ items and a knapsack with a fixed capacity. Each of these items has a weight and a profit value. This problem aims to put some items in the knapsack that give maximum profit so that the total weight of the items does not exceed the knapsack capacity. R Bellman~\cite{bellman1954some} first introduced this problem in 1954. The problem was proven to be \NPC by R.M.Karp in 1972~\cite{karp2009reducibility}. Later on, much more work has been done on the generalized version of this problem. ~\cite{martello1990knapsack,kellerer2004multidimensional,CacchianiILM22,cacchiani2022knapsack}.

In our problem, we include a graph constraint on \kp, where we represent each item as a vertex of a graph, and each vertex has weight and profit values. Based on a graph constraint, we choose some vertices that give the maximum profit within the fixed capacity of the knapsack. There are some existing variants of the \kp with graph constraints. Pferschy and Schauer showed  \NP-hardness and approximation results for some special graph classes for the \kp with the independent set as graph property~\cite{PferschyS09}. Dey et al. studied \kp with graph connectivity, path, shortest path~\cite{dey2024knapsack}, and \vc~\cite{dey2024knapsackwith}.

This paper considers another graph property as a constraint, the \dos~\cite{DBLP:conf/stoc/GareyJS74}, with the \kp. A \dos of an undirected graph is a set of vertices such that all other vertices are adjacent to at least one element of that set. We use the name \dosknapsack for our problem, which is formally defined in Section~\ref{sec:prob-def}. As a brief introduction, suppose $\GG(\VV,\EE)$ is a graph with each vertex $v\in\VV$ having weight $w(v)$ and profit $\alpha(v)$. The positive numbers $s$ and $d$ are the knapsack capacity and target profit, respectively. We want to find $\WW\subseteq\VV$ such that (I) $\WW$ is a \dos (II) $\sum_{v\in \WW}w(v)\leq s$ and (III) $\sum_{v\in \WW}\alpha(v)\geq d$.

Towards the motivation: \dosknapsack is a problem that combines two well-known problems. One is \kp on graphs and the other is \dos. Now, one can argue on the natural question that if, for a graph instance, both \dos and \kp are \yes instances, then obviously \dosknapsack is a \yes instance. However, it is not true in general. Consider the following instances of graph classes in Figure~\ref{star} and Figure~\ref{split}. We refer the reader to take a look at Definition~\ref{def:doskp} of \dosknapsack for a better understanding of the following examples.\\
\noindent%
\begin{minipage}[t]{.47\textwidth}
	\begin{figure}[H]
		\centering
		\includegraphics[height=2.9cm,width=2.9cm]{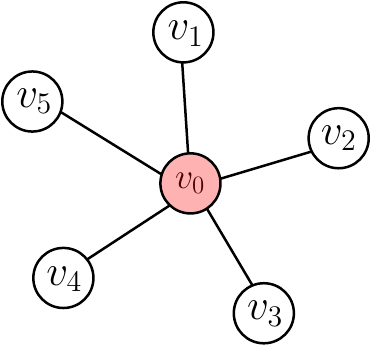}
		\caption{The instance of a star graph $\GG_1(\VV_1,\EE_1)$ with $\VV_1=\{v_0,\ldots, v_5\}$. $w(v_i)=\alpha(v_i)=1~ \forall i= 1, \ldots 5$ and $w(v_0)=\alpha(v_0)=5$ and $s=d=4$.}
		\label{star}
	\end{figure}
\end{minipage}%
\raisebox{-.22\textheight}{\rule{0.5pt}{.23\textheight}}
\hfill
\begin{minipage}[t]{.5\textwidth}
	\begin{figure}[H]
		\centering
		\includegraphics[height=2.4cm,width=4cm]{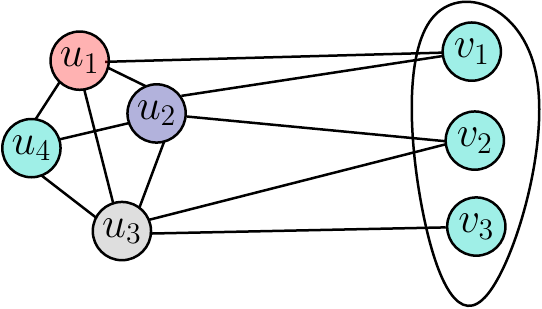}
		\caption{The instance of a split graph $\GG_2(\VV_2,\EE_2)$ with $\VV_2=\{u_i:1\leq i\leq 4\}\cup \{v_j:1\leq j\leq 3\}$. $w(u_i)=4, \alpha(u_i)=0~ \forall i= 1, \ldots 4$ and $w(v_j)=\alpha(v_j)=1~ \forall j= 1, \ldots 3$ and $s=d=3$. }
		\label{split}
	\end{figure}
	
\end{minipage}%

For both of these graph instances, \dos as well as \kp are \yes instances, but \dosknapsack is not. Moreover, this holds for trees and also for the perfect graph classes, as the star graphs (or stars), which are trees with a specific structure, and the split graphs, which represent the perfect graph class, as illustrated in the above graph instances.

Consider two other graph classes: one is a bipartite graph, and the other is a regular graph. Also, for the following instances in Figure~\ref{bipartite} and Figure~\ref{peterson}, \dosknapsack is \no instance even for the \yes instances of both \dos and \kp.\\
\noindent%
\begin{minipage}[t]{.47\textwidth}
	\begin{figure}[H]
		\centering
		\includegraphics[height=2.5cm,width=4cm]{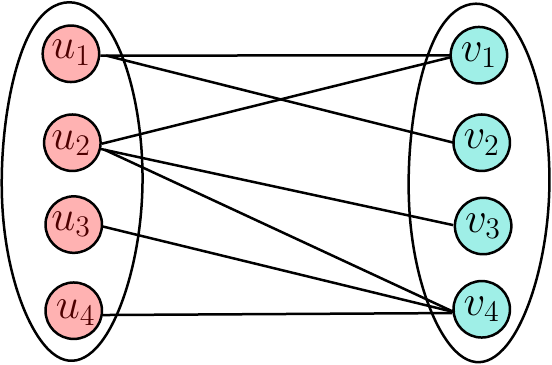}
		\newline
		\newline
		\caption{The instance of a bipartite graph $\GG_3(\VV_3,\EE_3)$ with $\VV_3=\{u_i:1\leq i\leq 4\}\cup \{v_j:1\leq j\leq 4\}$. $w(u_i)=4, \alpha(u_i)=0~ \forall~ i= 1, \ldots 4$ and $w(v_j)=\alpha(v_j)=1~ \forall~ j= 1, \ldots 3$ and $s=d=3$.}
		\label{bipartite}
	\end{figure}
\end{minipage}%
\raisebox{-.25\textheight}{\rule{0.5pt}{.23\textheight}}
\hfill
\begin{minipage}[t]{.5\textwidth}
	\begin{figure}[H]
		\centering
		\includegraphics[height=3cm,width=3.3cm]{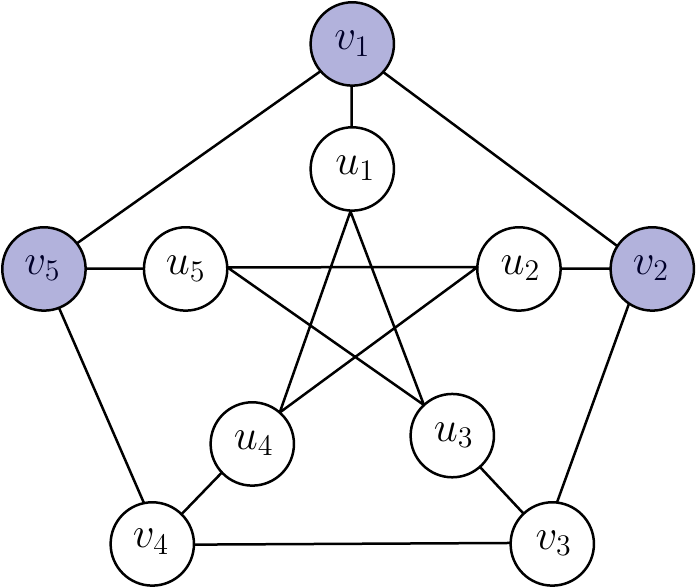}
		\caption{The instance of a regular graph $\GG_4(\VV_4,\EE_4)$ with $\VV_4=\{u_i:1\leq i\leq 5\}\cup \{v_j:1\leq j\leq 5\}$. $w(u_i)=w(v_j)=4, \alpha(u_i)=\alpha(v_j)=0~ \forall i= 1, \ldots 5, j=4, 5$ and $w(v_j)=\alpha(v_j)=1~ \forall j= 1, \ldots 3$ and $s=d=3$. }
		\label{peterson}
	\end{figure}
	
\end{minipage}%

Thus, it shows the level of difficulty of \dosknapsack compared with \kp and \dos. So, it is very clear from the above graph class instances that the \yes instance is not obvious for \dosknapsack even though both \kp and \dos are \yes instances. 

Now, consider a real-world scenario as a representation of our problem. 
Nowadays, we are all surrounded by different kinds of social networking platforms. These platforms are not only for acquiring knowledge or communicating with the rest of the world, but also a medium of business. One person is connected with many other people through this network. Some people broadcast different topics on their social networking channels, and based on their publicity, they earn some subscribers. In this scenario, suppose a commercial company wants to advertise on some of these broadcasting channels to promote their product, so that their advertisement reaches the most users in a cost-efficient manner. For this purpose, they pay a certain price to those channels for playing their advertisement, based on the number of subscriber or the length of their advertisement. Now, we can represent this problem with respect to the company by a graph-theoretic modeling, where each user (channel owner or subscriber) represents a vertex of the network and is connected via an edge to a subscriber. However, there is no direct edge between two subscribers. The goal of this company is to bring product information to as many users' knowledge as possible within their fixed budget capacity. We can represent this above-mentioned realistic problem in the form of \dosknapsack.

Consider another natural example: Suppose the government wants to optimally build power stations in different localities of a city so that each consumer of every locality can get an electric supply. Here, we consider each locality as the vertex. The power station of each locality has a finite range within which it can distribute electricity. If two such regions overlap, we draw an edge between those localities. Building a power station with different capacities, availability of places in a locality, etc., are the cost factors for this problem. The revenue generated by the total unit of power distributed to a locality to provide electricity in different houses and factories is the profit of this problem. Our problem \dosknapsack represents such a real-world example in the graph theoretic approach. Here, we want to find a subset of vertices to distribute power supply to the entire city and make optimal profit. Our problem can be widely applicable in different network problems like wireless ad-hoc networks~\cite{markarian2012degree}, railway transportation networks~\cite{weihe1998covering} (where stations are represented as the vertices, the connection between two stations through the railway tracks is denoted as edges, and we want to find a subset of stations to pass express trains), etc.

In Section~\ref{Sec:Cont}, we summarize our main contributions. We discussed the literature related to our work in Section~\ref{sec:r-work}. Some preliminary definitions used in other sections are given in Section~\ref{sec:prelim}. Our problem statement is defined in Section~\ref{sec:prob-def}. The results we proved on \NP-hardness are discussed in Section~\ref{Sec: Res-C-NPH}. In Section~\ref{Se: Res-Param-Comp}, we show the parameterized complexity results for our problem \dosknapsack. Finally, we conclude in Section~\ref{sec:conclusion}.


\section{Our Contributions}\label{Sec:Cont}
The main contributions of this paper are on Classical \NP-Hardness and Parameterized Complexity.\\

	\noindent
	\ding{118} \textbf{Classical \NP-Hardness Results:}\\
	
	We prove \dosknapsack and \kdosknapsack are strongly \NPC for general graphs, even for general bipartite graphs. Also, \minimaldosknapsack is strongly \NPC for general graphs. We show that though the star graph is a special case of a bipartite graph, \dosknapsack and \kdosknapsack are weakly \NPC for the star graphs. We provide a dynamic programming-based algorithm that solves \dosknapsack for trees in Pseudo-polynomial time. The running time of this algorithm is $\OO(n \cdot min\{s^2,(\alpha(V))^2\})$, where $s$ is the size of the knapsack and $\alpha(V)=\sum_{v\in\VV}\alpha(v)$.\\
	
	\noindent
	\ding{118} \textbf{Parameterized Complexity Results:}\\
	\begin{itemize}
		
		\item[$\bullet$] \textbf{Solution size:} We prove \dosknapsack is \WTw-hard by the reduction made from the parameterized version of \dos. By reducing the parameterized version of \dosknapsack to \wcs(\WCS), we prove that \dosknapsack belongs to \WTh. Thus, it is \WTh-complete. We show that unless \ExpoTH (\ETH) fails, there is no $f(k)n^{o(k)}$ -time algorithm parameterized by the solution size $k$. Similar results occur for \kdosknapsack.\\
		
		\item[$\bullet$] \textbf{Solution size of \vcknapsack:} \dosknapsack is Fixed-Parameter Tractable (\FPT) parameterized by the size of \vcknapsack solution. We prove this theorem by using the dynamic programming algorithm described in Theorem~\ref{thm:doskp-trees-pseudo-poly}. The running time of the algorithm is $\OO(2^{vck-1}\cdot n^{\OO(1)} \min\{s^2,{\alpha(\VV)}^2\})$, where $vck$ is the solution size of the \vcknapsack, $s$ is the size of the knapsack and $\alpha(\VV)=\sum_{v\in\VV}\alpha(v)$.\\
		
		\item[$\bullet$] \textbf{Treewidth:} By considering \tw as the parameter, we prove \dosknapsack is \FPT. We develop an algorithm inspired by the approach discussed in ~\cite{alber2002fixed}. The running time of the algorithm is $\OO(4^{tw}\cdot n^{\OO(1)} \min\{s^2,{\alpha(\VV)}^2\})$, where $tw$ represents the \tw of the given graph.The similar result occurs for \kdosknapsack and \minimaldosknapsack\\
		
	\end{itemize}

\section{Related Works}\label{sec:r-work}

Our problem is the generalization of Weighted \dos problem (WDS), which is \NPC ~\cite{maw1997weighted}. It is polynomial time solvable for trees~\cite{natarajan1978optimum}. For strongly chordal graphs, WDS is also polynomial time solvable  ~\cite{farber1984domination}. WDS is solvable in time $\OO(n^3)$~\cite{farber1985domination} for permutation graphs. There are algorithms that solve WDS in time $\OO(n)$ for interval graphs, and for circular-arc graphs, in time $\OO(n+m)$~\cite{chang1998efficient}. Minimum Weight \dos problem (MWDS) is \NP-Hard. There is a greedy heuristic algorithm~\cite{lin2016effective} for Minimum Weight \dos MWDS. Most existing approximation algorithms focus on some special graphs, such as disk graphs, growth-bounded graphs, etc. However, they are not easily generalized to general graphs. An polynomial time approximation algorithm (\PTAS)~\cite{wang2012ptas} (restricted to polynomial growth bounded graphs with a bounded degree constraint) and metaheuristic algorithms (hybrid genetic algorithm and hybrid ant colony optimization)~\cite{potluri2013hybrid} were introduced for MWDS.

The \dos problem is \FPT, where \tw is the parameter. If a graph $\GG$ containing $n$ vertices with tree decomposition of width at most $tw$ is given, then \dos problem can be solved in $4^{tw}\cdot {tw}^{\OO(1)}\cdot n$ and $3^{tw}\cdot n^{\OO(1)}$ time by two different algorithms~\cite{cygan2015parameterized}. By assuming Strong Exponential Time Hypothesis (\SETH) is \true, \dos can not be solved in time $(3-\epsilon)^{tw}\cdot n^{\OO(1)}$~\cite{lokshtanov2011known}. It is also \FPT in the graph that does not contain the complete bipartite graph $K_{(1,3)}$ as an induced subgraph. ~\cite{cygan2011dominating,hermelin2011domination}. There is an algorithm which determines whether a graph admits \dos of fixed cardinality $k \geq 2$ in time $\OO(n^{k+o(1)})$~\cite{eisenbrand2004complexity}. There is an \FPT algorithm for Planar \dos with complexity $\OO(8^k \cdot n)$ ~\cite{downey2013fundamentals,alber2005refined}. $\HH-$minor free graph $\GG$ with $n$ vertices admits polynomial kernel~\cite{fomin2012linear}. \dos problem is \WO-hard for the special class of graphs, like, axis-parallel unit squares, intersection graphs of axis-parallel rectangles, axis-parallel segments~\cite{markarian2012degree}. \dos problem is \WTw-complete, thus it is unlikely to have an \FPT algorithm for graphs in general~\cite{niedermeier2006invitation,cygan2015parameterized}. Constant factor approximation exists for planar graphs to find \dos in polynomial time ~\cite{baker1994approximation}. Till now, no constant factor approximation is known for \dos for general graphs~\cite{gambosi1999complexity}. The exact algorithm for the \dos problem runs in time $\OO(1.4969^n)$~\cite{van2011exact}, which is the best known result. The algorithm for Minimum \dos takes time $\OO(1.4864^n)$ with polynomial space~\cite{iwata2012faster}.


\section{Preliminaries}\label{sec:prelim}
This section provides a recap of key definitions and notations used throughout this paper. We define $\GG(\VV, \EE)$ as the graph with the set of vertices $\VV$ and the set of edges $\EE$, where $|\VV|=n$ and $|\EE|=m$. The size of the knapsack and the target value are represented as $s$ and $d$, respectively. The weight and profit of a vertex $u\in\VV$ are denoted as $w(u)$ and $\alpha(u)$ respectively. We denote the closed interval $[1,n]=\{x\in\mathbb{Z}:1\leq x \leq n\}$.
	\begin{definition}[\kp]\label{def:kp}
		Given a set $\XX=[1,n]$ of $n$ items with sizes $\theta_1,\ldots,\theta_n,$ values $p_1,\ldots,p_n$, capacity $b$ and target value $q$, compute if there exists a subset $\II \subseteq \XX$ such that $\sum_{i\in \II} \theta_i \leq b$ and $\sum_{i\in \II} p_i \geq q$. We denote an arbitrary instance of \kp by $(\XX,(\theta_i)_{i\in\XX}, (p_i)_{i\in\XX}, b,q)$.
	\end{definition}
	
	\begin{definition}[\dos]\label{def:dos}
		Given a graph $\GG=(\VV,\EE)$ and a positive integer $k$, compute if there exists a subset $\VV' \subseteq \VV$ such that any vertex of $\VV$ is either included in $\VV'$ or adjacent to at least one of the vertex of $\VV'$ and $|\VV'|\leq k$. We denote an arbitrary instance of \dos by $(\GG,k)$.
	\end{definition}
	
	\begin{definition}[\vc]\label{def:vc}
		Given a graph $\GG=(\VV,\EE)$ and a positive integer $k$, compute if there exists a subset $\VV' \subseteq \VV$ such that at least one end point of every edge belongs to $\VV'$ and $|\VV'|\leq k$. We denote an arbitrary instance of \vc by $(\GG,k)$.
	\end{definition}
	

\begin{definition}[\vcknapsack\cite{dey2024knapsackwith}]\label{def:vckp}
	
	Given an undirected graph $\GG=(\VV,\EE)$,  weight and profit of vertices are defined as the functions $w:\VV \rightarrow \mathbb{Z}_{\geq 0} $ and $\alpha:\VV \rightarrow \mathbb{R}_{\geq 0} $ respectively, size $s(\in \mathbb{Z^+})$ of the knapsack and target value $d(\in \mathbb{R^+})$, compute if there exists a subset $\UU\subseteq\VV$ of vertices such that (I) \UU is a \vc, (II) $\sum_{u\in\UU} w(u) \le s$, (III) $\sum_{u\in\UU} \alpha(u) \ge d$.
	We denote an arbitrary instance of \vcknapsack by $(\GG,(w(u))_{u\in\VV},(\alpha(u))_{u\in\VV},s,d)$.
	
\end{definition}


\begin{definition}[\udos]\label{def:udos}
	Given a graph $\GG=(\VV,\EE)$ and a positive integer $k$, compute if there exists a subset $\VV' \subseteq \VV$ such that $\VV'$ is a \dos of size at least $k$ i.e., $|\VV'|\geq k$. We denote an arbitrary instance of ~\udos by $(\GG,k)$.
\end{definition}
Maximum cardinality among all Minimal \dos problem or Upper Domination (UD) (denoted as $\Gamma(\GG)$) problem is \NPC~\cite{cheston1990computational}, polynomial time solvable for bipartite graphs~\cite{cockayne1981contributions,jacobson1990chordal} and Chordal graphs~\cite{jacobson1990chordal}. 



\begin{definition}[\ExpoTH]\label{def:eth}
	\ExpoTH(\ETH) states that there exists an algorithm solving the 3-SAT problem in time $\OO^*(2^{cn})$, where $\delta_3(>0)$ is the infimum of the set of constants $c$. Equivalently, there is no sub-exponential $(2^{o(n)})$time algorithm for the 3-SAT problem.
\end{definition}
\begin{definition}[\wcs]\label{def:wcs}
	Given a circuit $\CC$ and a positive integer $k$, compute if there exists a satisfying assignment with weight exactly $k$. We denote an arbitrary instance of \wcs(\WCS) as $(\CC,k)$.
\end{definition}
\begin{definition}[\wh]\label{def:w-hierarchy} 
	A parameterized problem \Pc is in the class $\mathsf{W}[t]$ if there is a parameterized reduction form \Pc to $\WCS_{(t,d)}$, where $t\geq 1$ for some $d\geq 1$. $d$ denotes the depth of the circuit and $t$ (weft) denotes the maximum number of large nodes (nodes with indegree $> 2$) on a path from the input to the output node.
\end{definition}


\section{\textbf{Our Problem Definitions}}\label{sec:prob-def}
We define our problems here formally. They are as follows:

\begin{definition}[\dosknapsack]\label{def:doskp}
	
	Given an undirected graph $\GG=(\VV,\EE)$,  weight and profit of vertices are defined as the functions $w:\VV \rightarrow \mathbb{Z}_{\geq 0} $ and $\alpha:\VV \rightarrow \mathbb{R}_{\geq 0} $ respectively, size $s(\in \mathbb{Z^+})$ of the knapsack and target value $d(\in \mathbb{R^+})$, compute if there exists a subset $\UU\subseteq\VV$ of vertices such that: 
		\begin{enumerate}
			\item \UU is a \dos,
			\item $\sum_{u\in\UU} w(u) \le s$,
			\item $\sum_{u\in\UU} \alpha(u) \ge d$.
	\end{enumerate}
	We denote an arbitrary instance of \dosknapsack by $(\GG,(w(u))_{u\in\VV},(\alpha(u))_{u\in\VV},s,d)$.
	
\end{definition}

\begin{definition}[\kdosknapsack]\label{def:kdoskp}
	
	Given an undirected graph $\GG=(\VV,\EE)$,  weight and profit of vertices are defined as the functions $w:\VV \rightarrow \mathbb{Z}_{\geq 0} $ and $\alpha:\VV \rightarrow \mathbb{R}_{\geq 0} $ respectively, size $s(\in \mathbb{Z^+})$ of the knapsack, target value $d(\in \mathbb{R^+})$ and $k \in  \mathbb{Z^+} $, compute if there exists a subset $\UU\subseteq\VV$ of vertices such that: 
		\begin{enumerate}
			\item \UU is a \dos of size $k$,
			\item $\sum_{u\in\UU} w(u) \le s$,
			\item $\sum_{u\in\UU} \alpha(u) \ge d$.
	\end{enumerate}
	We denote an arbitrary instance of \dosknapsack by $(\GG,(w(u))_{u\in\VV},(\alpha(u))_{u\in\VV},s,d,k)$.
	
\end{definition}

\begin{definition}[\minimaldosknapsack]\label{def:minimaldoskp}
	
	Given an undirected graph $\GG=(\VV,\EE)$,  weight and profit of vertices are defined as the functions $w:\VV \rightarrow \mathbb{Z}_{\geq 0} $ and $\alpha:\VV \rightarrow \mathbb{R}_{\geq 0} $ respectively, size $s(\in \mathbb{Z^+})$ of the knapsack, target value $d(\in \mathbb{R^+})$ and $k \in  \mathbb{Z^+} $, compute if there exists a subset $\UU\subseteq\VV$ of vertices such that: 
		\begin{enumerate}
			\item \UU is a Minimal \dos,
			\item $\sum_{u\in\UU} w(u) \le s$,
			\item $\sum_{u\in\UU} \alpha(u) \ge d$.
	\end{enumerate}
	We denote an arbitrary instance of \dosknapsack by $(\GG,(w(u))_{u\in\VV},(\alpha(u))_{u\in\VV},s,d,k)$.
	
\end{definition}



\section{\textbf{Results on Classical NP-hardness }}\label{Sec: Res-C-NPH}
This section presents the classical \NP-hardness results for our problem in various graph classes. \dos problem (see, Definition~\ref{def:dos}) is strongly \NPC as the reduction is done from \vc~\cite{DBLP:conf/stoc/GareyJS74}. For proving \NP-completeness of \dosknapsack, we reduce \dos to \dosknapsack. The theorem for strong \NP-completeness of \dosknapsack is shown below. 

\begin{theorem}\label{thm:doskp-npc}
	\dosknapsack is strongly \NPC.
\end{theorem}
	\begin{proof}
		Clearly, \dosknapsack $\in$ \NP. Let us construct the following instance $(\GG^\pr(\VV^\pr=\{u_i: i\in[1,n]\},\EE^\pr),(w(u))_{u\in\VV^\pr},(\alpha(u))_{u\in\VV^\pr},s,d)$ of \dosknapsack from a given arbitrary instance  $(\GG(\VV=\{v_i: i\in [1,n]\},\EE),k)$ of \dos of size $k$.
		\begin{align*}
			&\qquad\qquad\qquad\qquad\qquad\VV^\pr = \{u_i : v_i\in \VV, \forall i\in[1,n]\}\\
			&\qquad\qquad\qquad\qquad\qquad\EE^\pr = \{\{u_i,u_j\}: \{v_i, v_j\} \in\EE, i\neq j, \forall i, j\in[1,n]\}\\
			&\qquad\qquad\qquad\qquad\qquad w(u_i) = 1, \alpha(u_i)=1 \qquad \forall i\in[1,n]\\
			&\qquad\qquad\qquad\qquad\qquad s =d = k
		\end{align*}
		The \dosknapsack problem has a solution iff \dos has a solution.
		Let $(\GG^\pr,(w(u))_{u\in\VV^\pr},(\alpha(u))_{u\in\VV^\pr},s,d)$ be an instance of \dosknapsack such that $\WW^\pr$ be the resulting subset of $\VV^\pr$ with (i) $\WW^\pr$ is a \dos, (ii) $\sum_{u\in \WW^\pr} w(u) = k$, (iii) $\sum_{u\in \WW^\pr} \alpha(u) = k$. This means that the set $\WW^\pr$ is a \dos which gives the maximum profit $k$ for the bag capacity of size $k$. In other words, $\WW^\pr$ is a \dos of size $k$, as $\sum_{u\in \WW^\pr} w(u) = \sum_{u\in \WW^\pr} \alpha(u) = k$. Since $\WW' =\{u_i: v_i \in \WW, \forall i \in[1,n]\}$, $\WW$ is a \dos of size $k$. Therefore, the \dos instance is a \yes instance.
		
		Conversely, let us assume that \dos instance $(\GG,k)$ is a \yes instance. Then there exists a subset $\WW\subseteq \VV$  of size $k$ such that it outputs a \dos. Consider the set $\WW^\pr=\{u_i: v_i \in \WW, \forall i \in[1,n]\}$.
		Since each vertex of $\WW^\pr$ is involved with weight 1 and produces a profit amount 1, $\WW^\pr$ is a \dos of max bag size and total profit $k$. 
		\noindent
		Therefore, the \dosknapsack instance is a \yes instance.
	\end{proof}

This theorem is essential for proving the following theorem, which shows that \dosknapsack is strongly \NPC for a special graph class, namely, bipartite graphs. We reduce the \dosknapsack instance for the general graph to the \dosknapsack for bipartite graphs. 

\begin{theorem}\label{thm:doskp-bipt-npc}
	\dosknapsack is strongly \NPC even when the instances are restricted to a bipartite graph. 
\end{theorem}
	\begin{proof}
		We have already proven that \dosknapsack is \ NP-complete for general graphs. Now, we reduce an arbitrary instance of \dosknapsack $(\GG(\VV,\EE),(w(u))_{u\in\VV},(\alpha(u))_{u\in\VV}, s, d)$ for a general graph  to the instance of \dosknapsack  $(\BB(\hat{\VV},\hat{\EE}),(w(v))_{v\in\hat{\VV}},(\alpha(v))_{v\in\hat{\VV}}, \hat{s}, \hat{d})$ for a bipartite graph. The reduction is as follows:  
		\begin{align*}
			&\hat{\VV} = X \cup Y \cup \{z\}\\  
			&X=\{v_i : u_i\in \VV, \forall i\in[1,n]\}\\
			&Y=\{v_i' : u_i\in \VV, \forall i\in[1,n]\}\\
			&\hat{\EE} = \{\{v_i,v_i'\}\cup \{\{v_i, v_j'\}\cup \{v_i', v_j\}:\{u_i,u_j\} \in\EE\} \cup  \{z,v_i'\}: i\neq j, \forall i, j\in[1,n]\}\\
			&w(v_i)=0, \alpha(v_i)=0 \quad \forall i\in[1,n]\\
			&w(v_i')=w(u_i), \alpha(v_i')=\alpha(u_i)  \quad \forall i\in[1,n]\\
			&w(z) = 0, \alpha(z)=0\\
			&\hat{s}=s, \hat{d}=d
		\end{align*}
		
		Now, we can prove that the arbitrary graph $\GG$ has a solution for \dosknapsack iff the corresponding bipartite graph $\BB$ has a solution for \dosknapsack.
		
		Suppose, \dosknapsack is a \yes instance for the general graph $\GG$. That means $\GG$ has a \dosknapsack $\WW \subseteq \VV$ such that $\sum_{u\in \WW} w(u)\leq s$ and  $\sum_{u\in \WW} \alpha(u) \geq d$. Let $\hat{\WW}(\subseteq \hat{\VV})=\{v_i' : u_i\in \WW\}\cup \{z\}$  is the \dos in $\BB$, as $\{v_i' : u_i\in \WW\}$ dominates all the vertices of $X$ and $z$ dominates all the vertices of $Y$ and itself. Now,
		\[\sum_{v\in \hat{\WW}} w(v) = w(z) + \sum_{u\in\WW} w(u) \leq s =\hat{s}, \quad \text{since} \quad w(z) = 0\]
		\[\sum_{v\in \hat{\WW}} \alpha(v) = \alpha(z)+ \sum_{u\in\WW} \alpha(u) \geq d =\hat{d},\quad \text{since} \quad \alpha(z) = 0\]
		Hence, \dosknapsack is a \yes instance for bipartite graphs $B$.
		
		Conversely, let \dosknapsack is a \yes instance for bipartite graphs $B$ and $\hat{\WW}$ be any of its arbitrary solution such that $\sum_{v\in \hat{\WW}} w(v)\leq \hat{s}$ and  $\sum_{v\in \hat{\WW}} \alpha(v) \geq \hat{d}$.
		
		\noindent
		\textbf{\textul{Case 1}}: Let $z\in \WW'$.\\
		Without loss of generality, let $X\cap \hat{\WW} \neq \phi $ and/or $Y\cap \hat{\WW} \neq \phi $. Then choosing the vertices $u_i$ from $\VV$ of $\GG$ corresponding to $v_i\in X\cap \hat{\WW}$ and/or $v_i'\in Y\cap \hat{\WW}$ from $\BB$, for some $i\in[1,n]$. Let $\WW=\{u_i:v_i\in (X\cap \hat{\WW}) \text{~or~} v_i'  (Y\cap \hat{\WW}) \text{~or both}, \text{~for some~} i\in[1,n] \} \subseteq \VV$. Then,
		\begin{align*}
			\sum_{u_i\in\WW} w(u_i)&= \sum_{v_i\in (X\cap \hat{\WW})}w(v_i)+\sum_{v_i'\in (Y\cap \hat{\WW})}w(v_i'),\qquad\text{since}~ w(v_i)=0, w(v_i')=w(u_i)\\
			&= \sum_{v_i\in (X\cap \hat{\WW})}w(v_i)+\sum_{v_i'\in (Y\cap \hat{\WW})}w(v_i') + w(z), \qquad\text{as}~ w(z)=0\\
			&=\sum_{v\in \hat{\WW}} w(v) \leq \hat{s}  =s\\
			\sum_{u_i\in\WW} \alpha(u_i)&= \sum_{v_i\in (X\cap \hat{\WW})}\alpha(v_i)+\sum_{v_i'\in (Y\cap \hat{\WW})}\alpha(v_i'),\qquad\text{since}~ \alpha(v_i)=0, \alpha(v_i')=\alpha(u_i)\\
			&= \sum_{v_i\in (X\cap \hat{\WW})}\alpha(v_i)+\sum_{v_i'\in (Y\cap \hat{\WW})}\alpha(v_i') +\alpha(z), \qquad\text{as}~ \alpha(z)=0\\
			&=\sum_{v\in \hat{\WW}} \alpha(v) \geq \hat{d}  =d
		\end{align*}

		\noindent
		\textbf{\textul{Case 2}}: Let $z\notin \WW'$.\\
		This case also gives the same results as we omit the term $w(z)$ and $\alpha(z)$ from $\sum_{v\in \hat{\WW}} w(v)$ and $\sum_{v\in \hat{\WW}} \alpha(v)$, respectively.
		
		Thus, for both of these cases, $\sum_{u\in\WW} w(u) \leq s, \sum_{u\in\WW} \alpha(u) \geq d$, independent of whether $z\in\WW'$ or not. Therefore,  \dosknapsack is a \yes instance for the general graph $\GG$.
\end{proof}

The following theorem applies to star graphs, a special case of bipartite graphs. Though the star graph is bipartite, the following theorem proves that our problem is weakly \NPC for it. 

\begin{theorem}\label{thm:doskp-star-npc}
	\dosknapsack is \NPC even for star graphs.
\end{theorem}
	\begin{proof}
		\dosknapsack clearly belongs to \NP. We reduce \kp to \dosknapsack. Let $(\XX=[1,n],(\theta_i)_{i\in\XX}, (p_i)_{i\in\XX}, b,q)$ be an arbitrary instance of \kp and using it construct the instance $(\GG(\VV,\EE),(w(v))_{v\in\VV},(\alpha(v))_{v\in\VV},s,d)$ of \dosknapsack for star graphs as follows:\\
		\begin{align*}
			&\qquad\qquad\qquad\qquad\qquad \VV = \{v_i: i\in\XX=[1,n]\}\cup\{v_0\}\\
			&\qquad\qquad\qquad\qquad\qquad \EE = \{\{v_0,v_i\}: 1\le i\le n\}\\
			&\qquad\qquad\qquad\qquad\qquad w(v_0)=\alpha(v_0)=0;\\
			&\qquad\qquad\qquad\qquad\qquad w(v_i) = \theta_i , \alpha(v_i) = p_i~\forall i\in[1,n]\\ 
			&\qquad\qquad\qquad\qquad\qquad s = b, d = q
		\end{align*}
		
		We now claim that the two instances are equivalent, i.e., the \dosknapsack problem has a solution iff \kp has a solution.
		
		In one direction, let us suppose that the \kp instance is a \yes instance. Let $\II\subseteq\XX$ be a solution of \kp such that $\sum_{i\in\II} \theta_i \le b$ and $\sum_{i\in\II} p_i \ge q$. Consider $\UU=\{v_i: i\in\II\}\cup\{v_0\}\subseteq\VV$. We observe that $\UU$ is \dos, since $v_0\in\UU$. We also have
		\[\sum_{v\in \UU} w(v) = \sum_{i\in\II} w(v_i) = \sum_{i\in\II} \theta_i \le b = s, ~\text{and}~ \sum_{v\in \UU} \alpha(v) = \sum_{i\in\II} \alpha(v_i) = \sum_{i\in\II} p_i \ge q = d.\]
		Hence, the \dosknapsack instance is a \yes instance.
		
		In the other direction, let us assume that the \dosknapsack instance is a \yes instance with $\UU\subseteq\VV$ be a \dos such that $\sum_{v\in \UU} w(v) \le s$ and $\sum_{v\in \UU} \alpha(v) \ge d$. Let us consider a set $\II=\{i: i\in[n], v_i\in\UU\}$. We now have
		\[ \sum_{i\in\II} \theta_i = \sum_{i\in\II} w(v_i) =\sum_{v\in \UU} w(v) \le s = b,~\text{and}~ \sum_{i\in\II} p_i = \sum_{i\in\II} \alpha(v_i) = \sum_{v\in \UU} \alpha(v) \ge d=q.\]
		Hence, the \kp instance is a \yes instance.
\end{proof}
We proved that for graphs like star graphs with a particular tree structure, \dosknapsack is weakly \NPC. Now, we explore the results for trees. We provide a dynamic programming-based Pseudo-polynomial time algorithm for trees. The theorem is as follows:
\begin{theorem}\label{thm:doskp-trees-pseudo-poly}
	There is a pseudo-polynomial time algorithm of \dosknapsack for trees with complexity $\OO(n \cdot min\{s^2,(\alpha(V))^2\})$, where $s$ is the size of the knapsack and $\alpha(V)=\sum_{v\in\VV}\alpha(v)$. 
\end{theorem}
\begin{proof}
	Let $(\GG, (w(u))_{u\in\VV}, (\alpha(u))_{u\in\VV}, s, d)$ be an arbitrary instance of \dosknapsack. We develop a dynamic programming algorithm to compute \dosknapsack for trees \TT. Let \GG=\TT. This algorithm computes the set of "undominated weight-profit pairs" of Dominating Sets in $D[u, d_u]$ (including $u$) and $D[\bar{u}, d_u]$ (without including $u$), where $u$ is the root of the tree \TT. Then, we recursively compute from its next-level children, $v_1, v_2, \ldots, v_{d_u}$. We use the term "undominated weight-profit pair" as the pair $(w_c, \alpha_c)$ whose first component is the sum of weights $w_c = w(\CC) = \sum_{v\in \CC}w(v)$ and second component is the sum of profit values $\alpha_c = \alpha(\CC) = \sum_{v\in \CC}\alpha(v)$ for each vertex $v$ that belongs to the \dos $\CC$ such that $w(\CC)<w(\CC')$ or $\alpha(\CC)>\alpha(\CC')$ or, both compared with some other \dos $\CC'$, where $\CC$ and $\CC'$  both are the dominating sets for the node $u$ in $D[u, d_u]$. We Formally define $D[u, d_u]=\{(w_c, \alpha_c):\exists$ some $\dos$ $\CC$ with $w(\CC)=w_c<w(\CC')$ and/or $\alpha(\CC) = \alpha_c>\alpha(\CC')$, where $\CC'$ is another \dos and $u$ is included in both $\CC$ and $ \CC'\}$. $D[\bar{u}, d_u]$ also carries the same definition but without including $u$.

	\noindent%
	\begin{minipage}[t]{.63\textwidth}
		\parindent=.6cm
		We consider the maximum degree of $u$ is $d_u$ and  $v_1, v_2, \ldots, v_{d_u}$ with maximum degree $d_{v_1}, d_{v_2}, \ldots, d_{v_{d_u}}$, respectively. $D[u, i]$ (or $D[\bar{u}, i]$) is nothing but a 2-dimensional array of size $(i\times s)$ such that each cell contains a set of "undominated weight-profit pairs" including $u$ (or not including $u$), where $0\leq i\leq d_u$. So, it is obvious that  $D[u, d_u]$ and $D[\bar{u}, d_u]$ are 1-dimensional arrays and each cell of them contains a set of "undominated weight-profit pairs" of corresponding Dominating Sets.
		
	\end{minipage}%
	\raisebox{-.1\textheight}{\rule{0pt}{.1\textheight}}
	\hfill
	\begin{minipage}[t]{.35\textwidth}
		\begin{figure}[H]
			\centering
			\includegraphics[height=2cm,width=3cm]{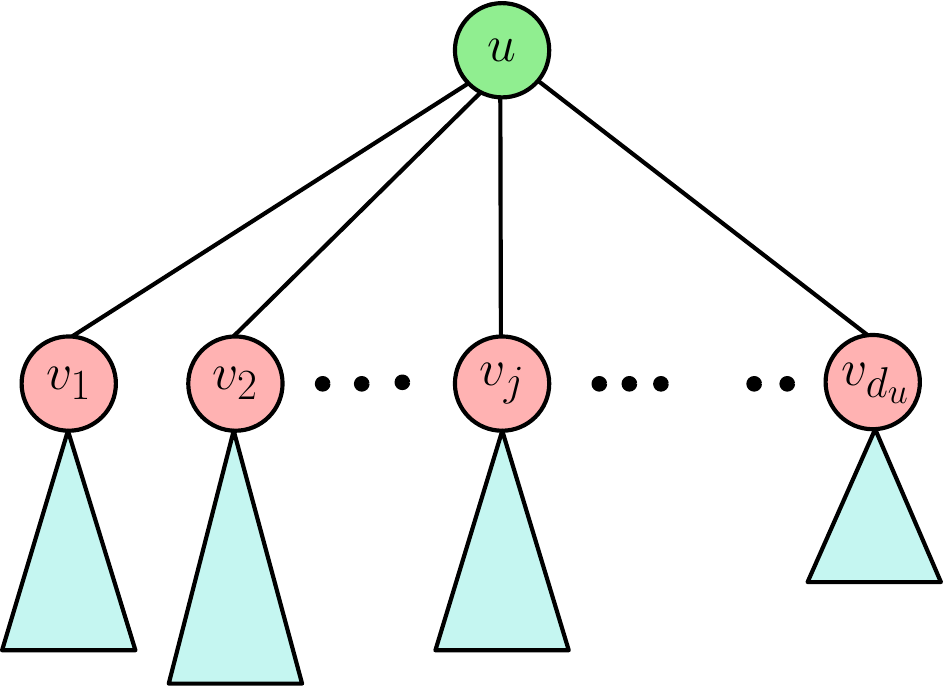}
			\caption{ Sub-trees for the tree rooted at $u$ and the children of $u$}
			\label{fig:t1}
		\end{figure}
	\end{minipage}%
	
	After computing $D[u, d_u]$ and $D[\bar{u}, d_u]$, we remove dominated and duplicated weight-profit pairs from $D[u, d_u] \cup D[\bar{u}, d_u]$ and check the condition of the knapsack size $s$ so that the knapsack can contain each undominated weight-profit pair of corresponding \dos individually. So, output \YES if a pair exists $(w_c, \alpha_c) \in D[u, d_u] \cup D[\bar{u}, d_u]$ such that $w_c \leq s$ and $\alpha_c \geq d$, otherwise \NO.\\
	
	\textbf{Initialization:} If $u$ is a leaf node $D[u, 0]=\{(w(u), \alpha(u))\}, D[\bar{u}, 0]=\{(0, 0)\}$. If $u$ is non-leaf node, we initialize $D[u, 0]=\{(w(u), \alpha(u))\}$, $D[\bar{u}, 0]=\{(0, 0)\}$, $D[\bar{u}, 1]=D[v_1, d_{v_1}]$, $P_v'=D[\bar{v_1}, d_{v_1}]$, where $P_v'$ is an 1-dimensional array containing all pairs of $D[\bar{v_1}, d_{v_1}]$ initially.\\
	
	\parindent=0cm 
	\textbf{Updating DP table:} Using bottom-up recursive computation, we have  $D[v_x,d_{v_x}]$ and $D[\bar{v_x},d_{v_x}]$ for each child $v_x$ of $u$ for $1\leq x\leq d_u$. Now, we compute $D[u, j]$ in $j^{th}$ iteration for $1\leq j\leq d_u$, where $D[u, j-1]$ and $D[\bar{u}, j-1]$ are already computed in $(j-1)^{th}$ iteration. Initially we set $D[u, j] = \emptyset$ and update $D[u, j]$ to $D[u, j]\cup\{(w_1 + w_2, \alpha_1 + \alpha_2)\}$ for each pair $(w_1, \alpha_1) \in D[v_j, d_{v_j}]\cup D[\bar{v_j}, d_{v_j}]$ and each pair $(w_2, \alpha_2) \in D[u, j-1]$ (See, Figure~\ref{fig:t1}). Then, we remove dominated and duplicated pairs from $D[u, j]$. This process executes for each value of $ j$ from $1$ to $d_u$. And, finally, we get $D[u, d_u]$.
	
	\parindent=.6cm
	Similarly, for computing $D[\bar{u}, j]$ for $2\leq j\leq d_u$ , we set $D[\bar{u}, j] = \emptyset$. But the difference here is that if the $j^{th}$ child $v_j$ of $u$ is a leaf node, we must have to include it in all the Dominating Sets because as $u$ is not included in any of the Dominating Sets of $D[\bar{u}, d_u]$, so $v_j$ would not be dominated by any of the vertices of $\TT$. Then update $D[\bar{u}, j]$ to $ D[\bar{u}, j]\cup\{(w_3 + w_4, \alpha_3 + \alpha_4)\}$ for each pair $(w_3, \alpha_3) \in D[v_j, d_{v_j}]$ and each pair $(w_4, \alpha_4) \in D[\bar{u}, j-1]$. Now, if $v_j$ is a non-leaf node, two situations are possible for updating  $D[\bar{u}, j]$. That is for computing Dominating Sets without including $u$,
	
	$\bullet$ (I) if in the $(j-1)^{th}$ iteration at least one of child $v_i (0\leq i\leq j-1)$ of $u$ dominates $u$, then we may or may not include $v_j$ in the Dominating Sets. So, we update $D[\bar{u},j]$ as $D[\bar{u}, j]\cup\{(w_5 + w_6, \alpha_5 + \alpha_6)\}$ for each pair $(w_5, \alpha_5) \in D[v_j, d_{v_j}]\cup D[\bar{v_j}, d_{v_j}]$ and each pair $(w_6, \alpha_6) \in D[\bar{u}, j-1]$ and also 
	
	$\bullet$ (II) if in the $(j-1)^{th}$ iteration no child $v_i (0\leq i\leq j-1)$ of $u$ dominates $u$, then we have to add $v_j$ in the Dominating Sets. To do so, we introduce another 1-dimensional array $P_v$ and initialize it with $\emptyset$ and update it with $P_v$ to $P_v\cup\{(w_7 + w_8, \alpha_7 + \alpha_8)\}$ for each pair $(w_7, \alpha_7) \in P_v'$ and each pair $(w_8, \alpha_8) \in D[v_j, d_{v_j}]$. $P_v'$ is actually contains all those pairs for which, till $(j-1)^{th}$ iteration $u$ is not dominated. Then we update $D[\bar{u}, j]$ to $ D[\bar{u}, j]\cup P_v$.
	
	These two steps (I) and (II) will execute sequentially. Now, we update $P_v'$ for the next iteration ($(j+1)^{th}$) by replacing all pairs of $P_v'$ by $ \{(w_9 + w_{10}, \alpha_9 + \alpha_{10})\} $ for each pair $(w_9, \alpha_9) \in P_v'$ and each pair $(w_{10}, \alpha_{10}) \in D[\bar{v_j}, d_{v_j}]$. Remove all dominated and duplicated pairs from $D[\bar{u}, j]$. We execute this process for all $j$ from $2$ to $d_u$. 
	After Computing $D[u, d_u]$ and $D[\bar{u}, d_u]$, we remove all dominated and duplicated pairs from $D[u, d_u] \cup D[\bar{u}, d_u]$ such that for each pair in $(w_c, \alpha_c) \in D[u, d_u] \cup D[\bar{u}, d_u]$, $w_c \leq s$ and $\alpha_c \geq d$ and return \YES if we set such pair $(w_c, \alpha_c)$ with $w_c \leq s$ and $\alpha_c \geq d$ in $D[u, d_u] \cup D[\bar{u}, d_u]$; Otherwise, \NO.\\
	
	\noindent
	{\textbf{Correctness:} Let after removal of all dominated and duplicated pairs from $D[u, d_u] \cup D[\bar{u}, d_u]$ such that $w_c \leq s$ and $\alpha_c \geq d$, we get the final set $\FF$ of undominated weight-profit pairs of Dominating Sets. Now we will prove that each pair $(w_c, \alpha_c) \in \FF$, is a \dos $\CC$ such that $w(\CC)=\sum_{v\in \CC}w(v)=w_c$ and $\alpha(\CC)=\sum_{v\in \CC}\alpha(v)=\alpha_c$.}
	
	Let us consider $\exists$ an arbitrary pair $(w_c, \alpha_c) \in \FF$. In the above algorithm, we store $ (w(v), \alpha(v))$ in $D[v, 0]$ if $v$ is leaf node. Since no vertex other than leaves is encountered in $0^{th}$ level, the above hypothesis holds trivially. Let us suppose $v_j$ is a non-leaf node, and it is the $j^{th}$ child of the root node $u$ and for all $j= 1$ to $d_u$, and $D[v_j, d_{v_j}] \cup D[\bar{v_j}, d_{v_j}]$ gives the sets of undominated weight-profit pairs of Dominating Sets at (root -1)$^{th}$ level. Now, for the root node $u$, we compute $D[u, d_u]$ and  $D[\bar{u}, d_u]$. From the above algorithm, each pair $(w_1, \alpha_1) \in D[v_j, d_{v_j}]\cup D[\bar{v_j}, d_{v_j}]$ is a \dos. Moreover, each pair $(w_2, \alpha_2) \in D[u, j-1]$ have calculated from $0^{th}$ to $(j-1)^{th}$ subtree of $u$ ($v_1$ to $v_{j-1}$ including $u$). Therefore, $(w_2, \alpha_2)$ is also a \dos, as in trees, there is no overlap between two subtrees. Since, we are summing up the pairs $(w_1+w_2, \alpha_1+\alpha_2)$ and storing it in $D[u, d_u]$, hence no pairs come twice in $D[u, d_u]$ and gives us \dos $\CC$ with $w(\CC)=w_c$ and $\alpha(\CC)=\alpha_c$ for some arbitrary pair $(w_c, \alpha_c)$. A similar thing will happen when we consider $D[\bar{u}, d_u]$. 
	
	So, it is clear to say that every pair is a \dos in $D[u, d_u] \cup D[\bar{u}, d_u]$ because $D[v_j, d_{v_j}]$ and $D[\bar{v_j}, d_{v_j}]$ gives the sets of undominated weight-profit pairs of Dominating Sets for all $j $ from $1$ to $d_u$. Thus, by considering the hypothesis TRUE in (root -1)$^{th}$ level, we proved by induction that it is also TRUE for the root node. \\
	
	\noindent
	\textbf{Complexity:} For each $n$ vertices we are computing $D[u, d_{u}]$ and $D[\bar{u}, d_{u}]$ for all $u \in \VV$. In computation of $D[u, d_{u}]$ we store $D[u, j]$ to $D[u, j]\cup\{(w_1+w_2, \alpha_1+\alpha_2)\}$ for each pair $(w_1, \alpha_1) \in D[v_j, d_{v_j}]\cup D[\bar{v_j}, d_{v_j}]$ and each pair $(w_2, \alpha_2) \in D[u, j-1]$. So, this part gives the complexity $\OO( min\{s^2, (\alpha(\VV))^2\})$ as each table entry can contain at most $min\{s, (\alpha(\VV))\}$ pairs and similarly for  $D[\bar{u}, d_{u}]$. And we compute $D[u, d_{u}]$, $D[\bar{u}$ and $d_{u}]$ once for each vertex. Hence the total complexity is $\OO(n\cdot min\{s^2, (\alpha(\VV))^2\})$.
\end{proof}
\kdosknapsack is a variant of \dosknapsack, where the size of the \dos is at most $k$. 

\begin{theorem}\label{thm:kdoskp-npc}
	\kdosknapsack is strongly \NPC. 
\end{theorem}
	\begin{proof}
		Clearly, \kdosknapsack $\in$ \NP. Let us construct the following instance $(\GG(\VV=\{u_i: i\in[1,n]\},\EE),(w(u))_{u\in\VV^\pr},(\alpha(u))_{u\in\VV^\pr},s,d, k)$ of \kdosknapsack from an given arbitrary instance  $(\GG'(\VV'=\{v_i: i\in [1,n]\},\EE'),k')$ of \dos of size $k'$.
		\begin{align*}
			&\qquad\qquad\qquad\qquad\qquad \VV = \{u_i : v_i\in \VV', \forall i\in[1,n]\}\\
			&\qquad\qquad\qquad\qquad\qquad \EE = \{\{u_i,u_j\}: \{v_i, v_j\} \in\EE', i\neq j, \forall i, j\in[1,n]\}\\
			&\qquad\qquad\qquad\qquad\qquad w(u_i) = 1, \alpha(u_i)=1 \qquad \forall i\in[1,n]\\
			&\qquad\qquad\qquad\qquad\qquad s =d = k =k'
		\end{align*}
		The \kdosknapsack problem has a solution iff \dos has a solution.
		
		Let $(\GG(\VV=\{u_i: i\in[1,n]\},\EE),(w(u))_{u\in\VV^\pr},(\alpha(u))_{u\in\VV^\pr},s,d, k)$ be the instance of \kdosknapsack such that $\WW$ be the resulting subset of $\VV$ with (i) $\WW$ is a \dos of size $k$, (ii) $\sum_{u\in \WW} w(u) \leq s = k $, (iii) $\sum_{u\in \WW} \alpha(u) \geq d = k$. This means that the set $\WW$ is a \dos which gives the maximum profit $k$ for the bag capacity of $k$. Now, since $k=k'$ and $\WW =\{u_i: v_i \in \WW', \forall i \in[n]\}$ , we can say that  $\WW'$ is a \dos of size $k$. Therefore, the \dos instance is a \yes instance.
		
		Conversely, let us assume that \dos instance $(\GG,k')$ is a \yes instance. Then there exists a subset $\WW'\subseteq \VV$  of size $k'$ such that it outputs a \dos.
		Since $\WW =\{u_i: v_i \in \WW', \forall i \in[n]\}$ and each vertex of $\WW$ is involved with weight 1 and produces profit amount 1 in the reduced graph $\GG$, $\WW$ is a \dos of max bag size and total profit $k'$. Also, we have $k'=k$.
		Therefore, the \kdosknapsack instance is a \yes instance.
\end{proof}
The \NP-completeness proof of this theorem is similar to the proof of Theorem~\ref{thm:doskp-npc} of \dosknapsack.

The proof of the following theorem differs slightly from the proof of \NP-completeness of \dosknapsack for bipartite graphs (Theorem~\ref{thm:doskp-bipt-npc}), as the size of the \dos is involved here.\\

\begin{theorem}\label{thm:kdoskp-bipt-npc}
	\kdosknapsack is strongly \NPC even when restricted to bipartite graphs. 
\end{theorem}
	\begin{proof}
	\dosknapsack is \NPC for the general graph even when we restrict the input graph class to bipartite, which we proved earlier. Now, we reduce an arbitrary instance of \kdosknapsack for the general graph $(\GG(\VV,\EE),(w(v))_{v\in\VV},(\alpha(v))_{v\in\VV}, s, d,k)$ of size $\leq k$ to the instance of a \kdosknapsack for bipartite graphs $(\BB(\hat{\VV},\hat{\EE}),(w(v))_{v\in\hat{\VV}},(\alpha(v))_{v\in\hat{\VV}}, \hat{s}, \hat{d}, k+1)$ of size $ \leq k+1$. The reduction is as follows:
	\begin{align*}
		&\qquad\qquad\qquad\qquad\qquad \hat{\VV} = X \cup Y\\  
		&\qquad\qquad\qquad\qquad\qquad X=\VV(=\{v_i : \forall i\in[1,n]\}) \cup \{x\}\\
		&\qquad\qquad\qquad\qquad\qquad Y=\VV'(=\{v_i' : v_i\in \VV, \forall i\in[1,n]\})\cup\{y\}\\
		&\qquad\qquad\qquad\qquad\qquad \hat{\EE} = \{\{v_i,v_i'\}\cup \{\{v_i, v_j'\}\cup \{v_i', v_j\}:\{v_i,v_j\} \in\EE\} \cup\\
		&\qquad\qquad\qquad\qquad\qquad \{x,v_i'\} \cup \{x,y\}: i\neq j, \forall i, j\in[1,n]\}\\
		&\qquad\qquad\qquad\qquad\qquad w(v_i)=0, \alpha(v_i)=0 \quad \forall i\in[1,n]\\
		&\qquad\qquad\qquad\qquad\qquad w(v_i')=w(u_i), \alpha(v_i')=\alpha(u_i)  \quad \forall i\in[1,n]\\
		&\qquad\qquad\qquad\qquad\qquad w(x)=w(y) = 0, \alpha(x)=\alpha(y)=0\\
		&\qquad\qquad\qquad\qquad\qquad \hat{s}=s, \hat{d}=d\\
	\end{align*}
	Now, we can prove that the arbitrary graph $\GG$ has a \kdosknapsack of size $\leq k$ iff the corresponding bipartite graph $\BB$ has a \kdosknapsack of size $\leq k+1$.\\
	\noindent
	Suppose $\GG$ has a \dosknapsack $\DD$ of size $\leq k$. So, $\sum_{u\in \DD} w(u)\leq s$ and  $\sum_{u\in \DD} \alpha(u)\geq d$. Now, either $\DD_1=\{x\}\cup \{v_i': u_i \in\DD\}$ or $\DD_2=\{y\}\cup \{v_i': u_i \in\DD\}$ must be the \dos for $\BB$. Because $\{v_i': u_i \in\DD\}$ dominates all the vertices $\VV$ and $\VV'$, and we need to choose any one of $x,y$ for the edge $\{x,y\}$. Let, $\hat{\DD}=\DD_1~\text{or}~\DD_2$ is the solution of \kdosknapsack for $\BB$. Thus, $B$ has \dosknapsack $\DD$ of size $\leq k+1$ with $\sum_{v\in \hat{\DD}} w(v)\leq \hat{s}$ and  $\sum_{v\in \hat{\DD}} \alpha(v)\geq \hat{d}$.\\
	
	Conversely, suppose $\BB$ has \dosknapsack $\hat{\DD}$ of size $\leq k+1$ such that $\sum_{v\in \hat{\DD}} w(v)\leq \hat{s}$ and  $\sum_{v\in \hat{\DD}} \alpha(v)\geq \hat{d}$. Let $\DD=\{u_i:v_i\in\hat{D}\cap\VV~||~v'_i\in\hat{D}\cap\VV'~||~ \text{both}\}$ forms the \dosknapsack of $\GG$ of size $\leq k$ with  $\sum_{u\in \DD} w(u)\leq s$ and  $\sum_{u\in \DD} \alpha(u)\geq d$ as $\hat{\DD}$ contains at least one of $x$, $y$ and $ w(x)=\alpha(x)=0, w(y) = \alpha(y)=0$.  
	\end{proof}

\begin{theorem}\label{thm:kdoskp-star-npc}
	\kdosknapsack is \NPC even for star graphs.
\end{theorem}
 The proof of this theorem is very similar to the Theorem~\ref{thm:doskp-star-npc}, so we skip this to avoid being redundant.

Now, for another variant \minimaldosknapsack of \dosknapsack, we prove the following theorem \NP-complete by reducing \udos (see Definition~\ref{def:udos}) to \minimaldosknapsack. 

\begin{theorem}\label{thm:minimaldoskp-npc}
	\minimaldosknapsack is strongly \NPC.
\end{theorem}
	\begin{proof}
		Clearly, \minimaldosknapsack $\in$ \NP. Since \udos is \NPC,  we reduce \udos to \minimaldosknapsack to prove NP-completeness. Let $(\GG(\VV=\{v_i: i\in[1, n]\},\EE),k)$ be an arbitrary instance of \udos of size $k$. We construct the following instance $(\GG^\pr(\VV^\pr=\{u_i: i\in[1, n]\},\EE^\pr),(w(u))_{u\in\VV^\pr},(\alpha(u))_{u\in\VV^\pr},s,d)$ of \minimaldosknapsack.
		\begin{align*}
			&\qquad\qquad\qquad\qquad\qquad \VV^\pr = \{u_i : v_i\in \VV, \forall i\in[1, n]\}\\
			&\qquad\qquad\qquad\qquad\qquad \EE^\pr = \{\{u_i,u_j\}: \{v_i, v_j\} \in\EE, i\neq j, \forall i, j\in[1,n]\}\\
			&\qquad\qquad\qquad\qquad\qquad w(u_i) = 1, w(u_j)=0 \quad \text{for~some~} i,j \in [1,n] \\
			&\qquad\qquad\qquad\qquad\qquad \alpha(u_i)=1 \quad \forall i \in [1,n] \\
			&\qquad\qquad\qquad\qquad\qquad s = d =k
		\end{align*}
		\parindent=.6cm
		The reduction of \udos to \minimaldosknapsack works in polynomial time. The \minimaldosknapsack problem has a solution iff \udos has a solution.
		
		Let $(\GG^\pr(\VV^\pr,\EE^\pr),(w(u))_{u\in\VV^\pr},(\alpha(u))_{u\in\VV^\pr},s,d)$ of \minimaldosknapsack such that $\WW$ be the resulting subset of $\VV'$ with (i) $\WW$ is a minimal \dos, (ii) $\sum_{u\in \WW} w(u) \leq s = k$, (iii) $\sum_{u\in \WW} \alpha(u) \geq d = k$. From the above condition (i), as $\WW$ is minimal \dos, we can not remove any vertex from $\WW$ which gives another \dos. Now, since we rewrite the second and third conditions as $|\WW| \leq k$ and $|\WW| \geq k$. So, $|\WW|$ must be $k$. Now, let $\II=\{v_i:u_i \in \WW,\forall i\in[n]\}$. Since $\WW$ is the Minimal \dos of size $k$, then $\II(\subseteq \VV)$ must be a \dos set of size at least $k$.Therefore, the \udos instance is a \yes instance.
		
		Conversely, let us assume that \udos instance $(\GG,k)$ is a \yes instance.\\ 
		Then there exists a subset $\II\subseteq \VV$  of size $ \geq k$ such that it outputs an \udos. Let $\WW=\{u_i:v_i \in \II,\forall i\in[n]\}$. Since $\II$ is an \udos of size $ \geq k$, then $\WW$ must be a Minimal \dos of size $k$. So, $\sum_{u\in \WW} w(u) \leq k =s$ and $\sum_{u\in \WW} \alpha(u) \geq k = d$.\\
		\noindent
		Therefore, the \minimaldosknapsack instance is a \yes instance.
\end{proof}

\section{\textbf{Results on Parameterized Complexity Classes }}\label{Se: Res-Param-Comp}
In this section, we study \dosknapsack under the lens of the parameterized complexity, where our input is restricted by specific parameters. 

\subsection{Solution size}\label{subsec:param-sol-size} 
The first parameter we consider here is the 'solution size'. We study the results for \dosknapsack parameterized by the 'solution size'. The parameter is represented by $k$. A well-known result for the \dos problem is as follows:

\begin{theorem}(~\cite{niedermeier2006invitation})\label{thm:dos-W[2]} 
	\dos is \WTw-complete parameterized by the solution size.
\end{theorem}
The following reduction is the \FPT reduction from \dos to \dosknapsack.

\begin{theorem}\label{thm:doskp-FPT-reduction} 
	There is a parameterized reduction from \dos to \dosknapsack. 
\end{theorem}
	\begin{proof}
		Let $(\GG(\VV=\{v_i: i\in [1,n]\},\EE),k)$ be an instance of parameterized problem \dos of size $k$. Now, we construct the instance $((\GG^\pr(\VV^\pr=\{u_i: i\in[1,n]\},\EE^\pr),(w(u))_{u\in\VV^\pr},(\alpha(u))_{u\in\VV^\pr},s,d),k')$ of  \dosknapsack parameterized by $k'$ as follows: 
		\begin{align*}
			&\qquad\qquad\qquad\qquad\qquad \VV^\pr = \{u_i : v_i\in \VV, \forall i\in[1,n]\}\\
			&\qquad\qquad\qquad\qquad\qquad \EE^\pr = \{\{u_i,u_j\}: \{v_i, v_j\} \in\EE, i\neq j, \forall i, j\in[1,n]\}\\
			&\qquad\qquad\qquad\qquad\qquad w(u_i) = 1, \alpha(u_i)=1 \qquad \forall i\in[1,n]\\
			&\qquad\qquad\qquad\qquad\qquad s =d = k'= k
		\end{align*}
		$\bullet$ This computation to obtain the output instance is performed in polynomial time.\\ 
		$\bullet$ The parameter $k'$ is upper bounded by some computable function $f$ of k as here, $f(k)=k$.\\
		$\bullet$ Now we need to show the \dosknapsack is a \yes-instance iff \dos is a \yes-instance.
		
		Let $((\GG^\pr(\VV^\pr=\{u_i: i\in[1,n]\},\EE^\pr),(w(u))_{u\in\VV^\pr},(\alpha(u))_{u\in\VV^\pr},s,d),k')$ be an instance of the parameterized problem \dosknapsack such that $\WW^\pr$ be the resulting subset of $\VV^\pr$ with (i) $\WW^\pr$ is a \dos of size $k'=k$, (ii) $\sum_{u\in \WW^\pr} w(u) = k'= k$, (iii) $\sum_{u\in \WW^\pr} \alpha(u) = k' = k$. This means that the set $\WW^\pr$ is a \dos which gives the maximum profit $k$ for the bag capacity of size $k$. In other words, $\WW^\pr$ is a \dos of size $k$, as $\sum_{u\in \WW^\pr} w(u) = \sum_{u\in \WW^\pr} \alpha(u) = k$. Since $\WW' =\{u_i: v_i \in \WW, \forall i \in[1,n]\}$, $\WW$ is a \dos of size $k$. Therefore, the \dos instance is a \yes instance.
		
		Conversely, let us assume that \dos instance $(\GG,k)$ is a \yes instance. Then there exists a subset $\WW\subseteq \VV$  of size $k$ such that it outputs a \dos. Consider the set $\WW^\pr=\{u_i: v_i \in \WW, \forall i \in[1,n]\}$. Since each vertex of $\WW^\pr$ is involved with weight 1 and produces a profit amount 1, $\WW^\pr$ is a \dos of size $k$ with max bag size and total profit $k$. 
		Therefore, the \dosknapsack instance is a \yes instance.
\end{proof}

This \FPT-reduction from \dos to \dosknapsack shows that \dosknapsack is at least as hard as \dos. Thus, we can write it as follows:
\begin{corollary}\label{thm:doskp-W[2]}
	\dosknapsack and \kdosknapsack are \WTw-hard parameterized by the solution size.
\end{corollary}

We define \wh in Definition~\ref{def:w-hierarchy}. Now, for showing the membership of \dosknapsack in \WTw, we reduce it to the \wcs (see Definition~\ref{def:wcs}). The reduction is done in the proof of the following theorem.

\begin{theorem}\label{thm:doskp-WCS-reduction}
	There is a parameterized reduction from \dosknapsack to \wcs, which shows that \dosknapsack is in \WTh. 
\end{theorem}
\begin{proof}
	Let $((\GG(\VV=\{u_i: i\in[1,n]\},\EE),(w(u))_{u\in\VV},(\alpha(u))_{u\in\VV},s,d),k)$ be the instance of  \dosknapsack parameterized by $k$ such that there exists a subset $\WW\subseteq\VV$ satisfies (i) $\WW$ is a \dos of size $k$, (ii) $\sum_{u\in \WW} w(u) \leq s$, (iii) $\sum_{u\in \WW} \alpha(u)\geq d$.\\ 
	
	\noindent%
	\begin{minipage}[t]{.25\textwidth}
		{\scriptsize
			\begin{align*}
				&\text{Now, we construct the instance} ~(\CC,k) ~\text{of}~\\  
				&\wcs (\WCS)~\text{as follows:}\\
				&\CC_{\text{Input}} = \{x_i : u_i\in \VV, \forall i\in[1,n]\}\\
				&\CC_{\text{Gates}} =\{D_i:i\in[1,n]\}\cup\{S_j:j\in[1,k]\}\cup\\
				&\qquad\qquad\{S_{out},D_{out},W_{out}, P_{out}, F_{out}\}\\
				&D_i= x_i\vee(\bigvee\limits_{u_b\in N[u_i]} x_b) \qquad \forall i\in[1,n]\\
				&D_{out}=\bigwedge\limits_{i=1}^n D_i\\
				&S_j=\left(\bigvee\limits_{i=1}^n x_{ij}\right)=\begin{cases}
					1 & \text{if}~ \#_{(x_{ij}=1)}=k\\ 0 & \text{otherwise}
				\end{cases} \quad \forall j\in[1,k] \\
				&S_{out}=\bigwedge\limits_{j=1}^k S_j \\
				&W_{out}= \left(\bigvee\limits_{i=1}^n x_{i}\right)=\begin{cases}
					1 & \text{if}~ \sum w{(x_i)}\leq s\\ 0 & \text{otherwise}
				\end{cases}\\
				&P_{out}= \left(\bigvee\limits_{i=1}^n x_{i}\right)=\begin{cases}
					1 & \text{if}~ \sum \alpha{(x_i)}\geq d\\ 0 & \text{otherwise}
				\end{cases}\\
				&F_{out}=D_{out}\wedge S_{out}\wedge W_{out} \wedge P_{out}\\
				&\CC_{\text{Output}} = F_{out}
		\end{align*}}
	\end{minipage}%
	\raisebox{-.43\textheight}{\rule{0.5pt}{.45\textheight}}
	\hfill
	\begin{minipage}[t]{.45\textwidth}
		$\bullet$ \textbf{Example}
		\begin{figure}[H]
			\centering
			\includegraphics[height=2cm,width=3.5cm]{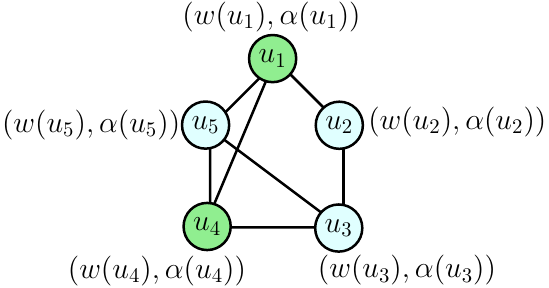}\\
			\includegraphics[height=4cm,width=5cm]{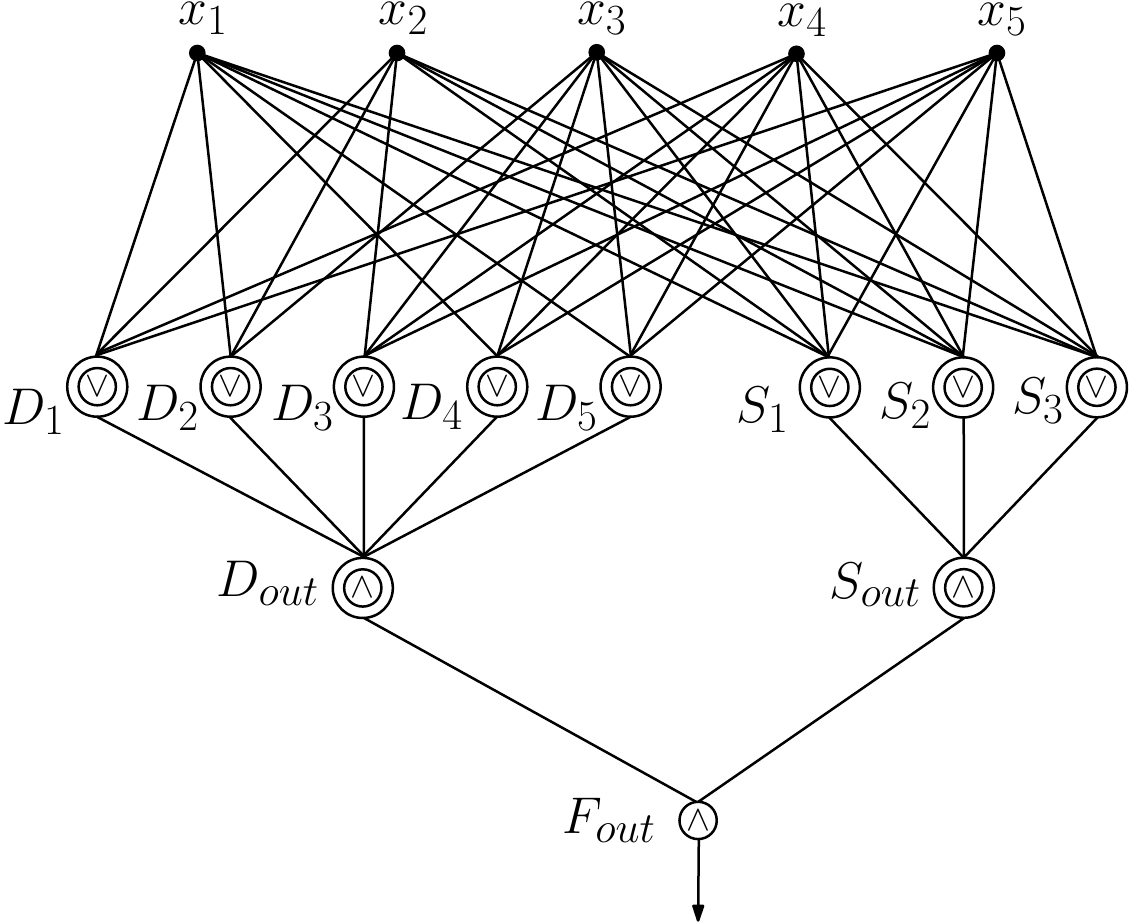}
			\caption{\scriptsize(a). The graph instance $(\GG,3)$ with $w(u_i) = \alpha(u_i)=1 ~\forall i\in[1,5]$ and and $s=d=3$. (b). The \WCS instance $(\CC,3)$  with \textit{wift}=2 corresponding to $\GG$ }
			\label{inputG and corresponding wcs}
		\end{figure}
	\end{minipage}%

	Here, $x_i$ is denoted as the input nodes for the circuit $\CC$ corresponding to the vertex $u_i$ of $\GG$, $\forall i\in [1,n]$. $x_b$ is the circuit input corresponding to the neighbour $u_b$ of $u_i$. If $D_{out}=1$, the corresponding vertices of circuit inputs with value 1 represent \dos, and $S_{out}=1$ gives the set of size $k$. To define $S_j$, we use $x_{ij}$ to identify $i$-th input in $j$-th clause. One additional condition for $S_j$ is that $S_j$ outputs $1$ only if for each $j$-th clause of disjunction, there must be a unique $i$ with exactly one $x_{ij}=1$ corresponding to $x_i=1$. In other words, the number-of $(x_{ij}=1)$ (\textit{viz.} $\#_{(x_{ij}=1)}$) for each clause $j\in [1,k]$ is exactly $k$ for corresponding $x_i=1$.   In this way we can get subset $\{x_{l_j}: \exists ~\text{unique}~ l\in[1,n] ~\text{for each}~ j\in[1,k]\}$ of size $k$ for $\{x_i\}_{i\in[1,n]}$. Similarly, $W_{out}$ and $P_{out}$ gives the output 1 if the conditions $\sum w{(x_i)}\leq s$ and $\sum \alpha{(x_i)}\geq d$ hold, respectively. $W_{out}$ and $P_{out}$ are represented as the total output weight and profit, respectively, corresponding to the \kp. We finally receive the output from the gate $F_{out}$.
	
	Now, for the yes-instance of \dosknapsack, the above circuit has a satisfying assignment, corresponding to the vertices $u \in \WW$. Conversely, for the satisfying assignment of the circuit, i.e., $F_{out}=1$, which means there exists a \dosknapsack of size $k$ with total weight $\leq s$ and profit $\geq d$.
	
	In the above construction of \WCS, one can find the maximum number of large nodes (nodes with indegree $> 2$) on a path from input to the output node (\textit{viz.} weft) is 3. Thus, we can say that \dosknapsack in \WTh.
\end{proof}
In the above example, the Figure~\ref{inputG and corresponding wcs} shows that if $w(u_i) = \alpha(u_i)=1 \forall i\in[1,n]$, then $s=d=k$. Thus, the following corollary holds according to Theorem~\ref{thm:dos-W[2]}.
\begin{corollary}\label{cor:doskp-W[2]-unit-weight-proft}
	\dosknapsack and \kdosknapsack are in \WTw for unary encoding of weights and profits. 
\end{corollary}

The theorem below immediately follows from Theorem~\ref{thm:doskp-FPT-reduction} and Theorem~\ref{thm:doskp-WCS-reduction}.
\begin{theorem}\label{thm:doskp-W[2]-Complete}
	\dosknapsack and \kdosknapsack are \WTh-complete parameterized by the solution size. 
\end{theorem}

\begin{observation}(~\cite{cygan2015parameterized})\label{obs:reduction-for-ETH}
	Suppose $A$, $B$ $\subseteq \sum^*\times \mathbb{N}$ be two problems with instances $(x,k)$ and $(x',k')$ respectively. There is a parameterized reduction from $A$ to $B$. Then $B$ has an $f(k)|x|^{\OO(g(k))}$-time algorithm implies an $f(k')|x|^{\OO(g(k'))}$-time algorithm for $A$ for some computable function $f$ and non-decreasing function $g$.
\end{observation}

We define \ExpoTH (\ETH) in Definition~\ref{def:eth} of Section~\ref{sec:prelim}. The following theorem is a conjecture on \dos, assuming \ETH.
\begin{theorem}(~\cite{cygan2015parameterized})\label{thm:dos-ETH}
	Assuming \ETH, there is no $f(k)n^{o(k)}$-time algorithm for the \dos problems, where f is any computable function.
\end{theorem}

We proved the following theorem by using Theorem~\ref{thm:dos-ETH}.

\begin{theorem}\label{thm:doskp-ETH}
	Unless \ETH fails, there is no $f(k)n^{o(k)}$ -time algorithm for \dosknapsack and \kdosknapsack for any non-decreasing function $f$ parameterized by the solution size k.
\end{theorem}
	\begin{proof}
		Let $(G,k)$ and $(G',k)$ be the instances of the \dos and \dosknapsack, respectively. From Theorem~\ref{thm:doskp-FPT-reduction} above, we have a parameterized reduction from \dos to \dosknapsack. Now, we prove this theorem by contradiction. Suppose \dosknapsack has a sub-exponential time algorithm with complexity $f(k)n^{o(k)}$. Now, from the Observation~\ref{obs:reduction-for-ETH} we can say that \dos problem also admits $f(k)n^{o(k)}$-time algorithm which contradicts the Theorem~\ref{thm:dos-ETH}.
\end{proof}
\subsection{Solution size of \vcknapsack}\label{subsec:param-vck} 
The next parameter we consider here is \vcknapsack (Definition~\ref{def:vckp}), which is defined in Section~\ref{sec:prelim}. For more results, we refer to the literature~\cite{dey2024knapsackwith}.

Here, we describe two simple observations for \vcknapsack.
\begin{observation}
	The solution of~\vcknapsack for connected graphs is also a solution of~\dosknapsack, but the converse is not true.
\end{observation}

\begin{observation}
	If we consider a graph with no edges, then the \vcknapsack is just nothing but the \kp. But in this situation, the \dosknapsack must choose all the vertices, and it has a solution if the sum of all weights of the vertices is $\leq s$. Otherwise, there is no solution. So, for the graph with no edges, \dosknapsack is solvable in polynomial time, whereas \vcknapsack is solvable in pseudo-polynomial time.
\end{observation}

The next theorem shows the running time of the algorithm for \dosknapsack where \vcknapsack is a parameter.
\begin{theorem}\label{thm:doskp-vc_kp}
	There is an algorithm of \dosknapsack with running time $\OO(2^{vck-1}\cdot n^{\OO(1)} \min\{s^2,{\alpha(\VV)}^2\})$ parameterized by solution size of the \vcknapsack $vck$, where $s$ is the size of the knapsack and $\alpha(\VV)=\sum_{v\in\VV}\alpha(v)$.
\end{theorem}
\begin{proof}
	Let $(\GG=(\VV,\EE),(w(u))_{u\in\VV},(\alpha(u))_{u\in\VV},s,d)$ be the given instance of \dosknapsack and $\SS\subseteq\VV$ with $|\SS|=vck$ be the given solution on the same instance for the \vcknapsack that satisfies the following conditions: $\sum_{u\in\SS} w(u) \le s , \sum_{u\in\SS} \alpha(u) \ge d$. Since $\SS$ is a \vc of $\GG$, it covers all the edges $\EE$ of $\GG$. All the vertices $\VV\setminus\SS$ are connected with any of the vertices of $\SS$ just like pendant or non-pendant vertices (See, Figure~\ref{vck}). So, we can consider this structure as a tree with $\SS$  as the root node with $t(=|\VV|-vck)$ number of leaf nodes. Now, we use a dynamic programming approach by calling the pseudo-polynomial time Algorithm~\ref{thm:doskp-trees-pseudo-poly} for every subset $\SS_c$ of $\SS$ by considering $\SS_c$ as the root and $\VV\setminus\SS$ as the children. 
	\begin{minipage}[t]{.75\textwidth}
		The total number of subsets of $\SS$ is $2^{vck}$ and  We compute dp tables $D[\SS_c, t]$ and $D[\bar{\SS_c}, t]$ for each subset $\SS_c$ of $\SS$, where $\bar{\SS_c}=\SS\setminus\SS_c$. So, we need to call the Algorithm~\ref{thm:doskp-trees-pseudo-poly} in $\frac{2^{vck}}{2}=2^{vck-1}$ times. The running time of the Algorithm~\ref{thm:doskp-trees-pseudo-poly} is $\OO(n \min\{s^2,{\alpha(\VV)}^2\})$. One little modification here is that before adding each weight-profit pair in the dp tables  $D[\SS_c, t]$ or $D[\bar{\SS_c}, t]$, we need to verify that the associated vertices corresponding to the weight-profit pair must form a \dos of $\GG$. This verification takes time $\OO(n)$. Each dp table $D[\SS_c, t]$ or $D[\bar{\SS_c}, t]$ can contain at most $\min\{s,{\alpha(\VV)}\}$ weight-profit pairs. In the final step we take union of $D[\SS_c, t]\cup D[\bar{\SS_c}, t]$ for all  subset $\SS_c$ of $\SS$.
	\end{minipage}%
	\raisebox{-.1\textheight}{\rule{0pt}{.1\textheight}}
	\hfill
	\begin{minipage}[t]{.20\textwidth}\label{vck}
		\begin{figure}[H]
			\centering
			\includegraphics[height=1.5cm,width=2.5cm]{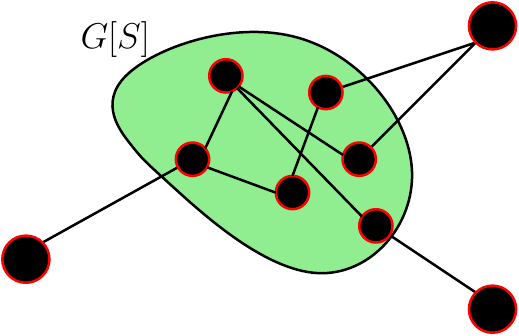}
			\caption{Subgraph $\GG[\SS]$ induced by the solution of \vcknapsack}
			\label{vck}
		\end{figure}
	\end{minipage}%
	
	After removing all dominated and duplicated weight-profit pairs we get the resulting pair $(w_c,\alpha_c)$, where $\mathlarger{\max}_{\alpha_c\geq \alpha} ~\{\mathlarger{\bigcup}_{\substack{ (w,\alpha)\in D[\SS_c, t]\cup D[\bar{\SS_c}, t]\\\forall\SS_c\subseteq\SS}} (w,\alpha) \}$. Thus the total running time is $\OO(2^{vck-1}\cdot n^{\OO(1)} \min\{s^2,{\alpha(\VV)}^2\})$.
\end{proof}
\begin{corollary}
	\dosknapsack is \FPT parameterized by the solution size of \vcknapsack.
\end{corollary}

\subsection{Treewidth}\label{subsec:param-tw}
In this subsection, we consider \tw as the parameter. We define the \td, \ntd, and \tw as follows:
	\begin{definition}[\td]\label{def:td}
		Given a graph $\GG=(\VV,\EE)$. The set of vertices $\chi=\{X_i\}_{i\in V}$, where each $X_i\subseteq \VV$ is called a bag such that $T=(V, E)$ is a tree structure for the graph $\GG$ with $i$ as the nodes of $~T$. We define a \td $T^\chi$ of $~\GG$ having the following properties:
		\begin{itemize}
			\item[$\bullet$] $\bigcup\limits_{i\in V}X_i= \VV$.
			\item[$\bullet$] For each edge $e=\{v_1,v_2\} \in \EE$ there must be a bag $X_k$ such that $\{v_1,v_2\} \subseteq X_k$, where $k=|V|$.
			\item[$\bullet$] If a node $j$ lies on the path between the nodes $i$ and $k$, then $X_i \cap X_k\subseteq X_j $, which means that the set $\{i, j, k\}$ induces a connected subtree of $~T$,  $\forall~ i, j, k \in V$.
		\end{itemize}
		We denote an arbitrary instance of ~\td $T^\chi$ of~ $~\GG$ by the pair $\langle \chi, T \rangle$. The \textsc{width} of~ $T^\chi$ is equal to $\max_{i\in V}|X_i| -1$ and the \tw$(tw)$ holds the minimum width value among all $T^\chi$ of~ $\GG$. We use the term 'vertices' for the vertices of the graph $\GG$ and 'nodes' for the vertices of $~T$.
	\end{definition}
	
	\begin{definition}[\ntd]\label{def:ntd}
		A \td $T^\chi_r$ rooted at a node $r$ is called a \ntd if the following properties holds:
		\begin{itemize}
			\item[$\bullet$] Every node of $~T$ has at most two children.
			\item[$\bullet$] A node with no children is called \textbf{leaf node}. 
			\item[$\bullet$] A node $i$ with one child $j$ is either called\\ 
			(i) \textbf{insert $\backslash$ introduce node} if ~$X_i=X_j\setminus\{v\} $ and $|X_i|=|X_j|+1$, or \\
			(ii) \textbf{forget node} if~ $X_j=X_i\setminus\{v\}$ and $|X_j|=|X_i|+1$ , where $v\in \VV$.
			\item[$\bullet$] A node $i$ with two children $j$ and $k$ is called \textbf{join node} if~ $X_i=X_j=X_k$.
		\end{itemize}
		
	\end{definition} 
The following lemma is the relation between the \td and \ntd.
\begin{lemma}\label{lm:ntd}(\cite{kloks1994treewidth})
	If the \td of width $k$ and $\OO(n)$ nodes are given for a graph $\GG$ with $n$ vertices, one can transform it into a \ntd of width  $k$ in time $\OO(n)$ having $\OO(n)$ number of nodes. 
\end{lemma}
Alber et al. proved that the \dos problem is \FPT parameterized by \tw ~\cite{alber2002fixed}. They provided an algorithm with running time $\OO(4^{tw}\cdot n)$. Later, using the Fast Subset Convolution method, the running time improved in $\OO(3^{tw}\cdot n^{\OO(1)})$~\cite{cygan2015parameterized}.

The following theorem presents another result for the \dosknapsack problem when we consider \tw as the parameter. 

\begin{theorem}\label{thm:doskp-treewidth}
	There is an algorithm of \dosknapsack with running time $\OO(4^{tw}\cdot n^{\OO(1)} \min\{s^2,{\alpha(\VV)}^2\})$ parameterized by \tw $tw$, where $s$ is the size of the knapsack and $\alpha(\VV)=\sum_{v\in\VV}\alpha(v)$.
\end{theorem}
\begin{proof}
	Let $((\GG,(w(u))_{u\in\VV},(\alpha(u))_{u\in\VV},s,d),tw)$ be the given instance of our problem and $T^\chi=\langle \chi, T \rangle$ be the \ntd  of $\GG$ with width at most $tw$ computed from a given \td of $\GG$ with width at most $tw$ in polynomial time, according to the Lemma~\ref{lm:ntd}.
	
	Suppose $\VV=\{v_1, v_2, \ldots, v_n\}$ is the set of vertices of the graph $\GG$. We represent $\GG_t=(\VV_t,\EE_t)$ as the subgraph corresponding to the tree decomposition $(T_t)$ rooted at $t$ such that $\VV(\GG) \cap V(T_t)=\VV_t(\GG_t)$, and $\EE_t$ are the edges induced by $\VV_t$. Now, a bag $X_i$ represented as $X_i=(v_{i_1}, v_{i_2},\ldots, v_{i_{n_i}})$ and the colors of the vertices of a bag with the following colors: 
	\begin{itemize}
		\item[$\bullet$] \textbf{Black:} These vertices are in the partial solution in $\GG_t$ that forms dominating set. In our algorithm, we represent these vertices as "1".
		\item[$\bullet$] \textbf{White:} These vertices are dominated by some vertices that are in partial solution. Thus, "white" vertices are not part of the partial solution. We represent these vertices with "0".
		\item[$\bullet$]\textbf{Gray:} These vertices are not in the partial solution and not even dominated at the current state of the algorithm. We represent "gray" vertices with "$\hat{0}$".
	\end{itemize}
	Here, we define a function $\tilde{c}:v_{i_t}\rightarrow \{1, 0, \hat{0}\}$ such that $\tilde{c}(v_{i_t})=c_{i_t}\in \{1,0,\hat{0}\}$. Thus, the vector $c=(c_{i_1}, c_{i_2},\ldots, c_{i_{n_i}}) \in \{1,0,\hat{0}\}^{n_i}$ of dimension $|X_i|=n_i$ corresponding to the bag $X_i=(v_{i_1}, v_{i_2},\ldots, v_{i_{n_i}})$ represents a  \textit{coloring} of the bag $X_i$. A coloring $\hat{c}=\{\hat{c}_{i_1}, \hat{c}_{i_2}, \ldots, \hat{c}_{i_{n_i}}\}$ leads to \textit{no valid solution} for the bag $X_i$ if for all vertex $v_x \in X_i$ with color $\hat{c}_x=0$ where $x\in\{i_1,i_2,\ldots,i_{n_i}\}$ such that  $\tilde{c}(n(v_{x})\cap \VV_i) \neq 1 \forall$  $ n(v_{x})\in N(v_{x})$. Here, $N(v_{x})= \{n(v_{x}): \{n(v_{x}), v_{x}\}\in \EE(\GG) \}$ is denoted as the set of neighbors of $v_{x}$ and  $\VV_i$ is the set of vertices of the partial solution in $\GG_i$ such that corresponding tree decomposition $T^\chi_i$ rooted at $i$.
	
	For finding the solution \dosknapsack, our goal is to find \dos along with the condition that it gives us optimal profit, i.e., maximum $\alpha(\VV)=\sum_{v\in\VV}\alpha(v)$ within our budget $s$ constraint. The existing \tw algorithm yields the minimum cardinality of \ dos, but it is possible that it may or may not yield the maximum profit. On the other hand, each vertex is associated with a weight, and our budget $s$ is fixed. So, it is not possible to add more vertices to the \kp to maximize the profit.
	
	We use dynamic programming and define the table $D$. The table entry $D[i,c]$ represents the set of undominated weight-profit pairs for the coloring $c$ and the node $i\in V(T)$ associated with the bag $X_i$. The algorithm chooses a node $i$ by applying DFS on $T^\chi_r$ rooted at $r$.
	\begin{itemize}
		\item[$\bullet$] If $i$ is \textit{leaf node}, the initialization for all the \textit{coloring} $c=(c_{i_1}, c_{i_2},\ldots, c_{i_{n_i}}) \in \{1,0,\hat{0}\}^{|X_i|}$, where  $X_i=(v_{i_1}, v_{i_2},\ldots, v_{i_{n_i}})$ :
		
		{\scriptsize\[D[i,c] = \begin{cases}
				\emptyset & \text{if $c$ loads to \textit{no valid solution}}\\
				\left(\mathlarger{\sum}_{\substack{c_j=1\\j\in \{i_1,\ldots, i_{n_i}\} }} w(v_j), \mathlarger{\sum}_{\substack{c_j=1\\j\in \{i_1,\ldots, i_{n_i}\} }} \alpha(v_j)\right) & \text{if $c$ valid and $\mathlarger{\sum}_{\substack{c_j=1\\j\in \{i_1,\ldots, i_{n_i}\} }} w(v_j)\leq s$}\\ \emptyset & \text{otherwise}
			\end{cases}\]
			
		}
		Here, computing summations for each $c\in {\{1,0,\hat{0}\}^{|X_i|}}$ executes $\OO(|X_i|)$ times. That is, in total $\OO(3^{|X_i|}\cdot |X_i|)$ time required to compute summations. Also, checking each weight-profit pair for all $D[i,c]$ takes $\OO(3^{|X_i|})$ times for all coloring $c$, as each $D[i,c]$ contains at most one such pair, in this step. Thus the total time for this initialization step takes $\OO(3^{|X_i|}\cdot |X_i|)$.\\
		
		\item[$\bullet$] If $i$ is \textit{forget node} with one child $j$:\\ 
		Let the bag corresponding to the node $i$ be $X_i$, and the child bag of $X_i$ is $X_j=X_i\cup\{x\}$. Suppose $c$ and $c'=(c\times\{x\})$ is the colorings for $X_i$ and $X_j$ respectively and the possible colors for $x:\tilde{c}(x)\in \{1,0,\hat{0}\}$. Then we updated $D[i,c]$ according to as follows: \[D[i,c]=\bigcup_{\tilde{c}(x)\in \{1,0\}} D[j,c']\] 
		We can refrain to gather the values of $D[j,c']$ for the color $\hat{0}$ of $x$ as according to the Definition~\ref{def:td} of \td that if a vertex $x$ which is unresolved $(\hat{0})$ and removed for a node at the current state, will never appear again for the rest of the algorithm to contribute a role to compute the \dos. Computation of $D[i,c]$ for all coloring $c$ takes $\OO(3^{|X_i|})$ times. Then we remove all \textit{undominated weight-profit pairs} from each $D[i,c]$ that contains at most $\min\{s,\alpha(\VV)\}$ pairs takes $\OO(3^{|X_i|}\cdot \min\{s,\alpha(\VV)\})$ times $\forall c\in \{1,0,\hat{0}\}^{|X_i|}$. Thus, the execution of the \textit{forget node} operation takes $\OO(3^{|X_i|}\cdot \min\{s,\alpha(\VV)\})$ times in total.\\
		
		\item[$\bullet$] If $i$ is \textit{introduce node} with one child $j$:\\ 
		The corresponding bags for the node $i$ and $j$ are $X_i$ and $X_j$ respectively and $X_i=X_j\cup \{x\}$. So, $X_j\subset X_i$. Let $c'$  and $c=(c'\times\{x\})$ are the colorings of $X_j$ and $X_i$ respectively. The coloring for $x$ are: $\tilde{c}(x)\in \{1,0,\hat{0}\}$. Then
		{\scriptsize\[D[i,c] = \begin{cases}
				D[j,c'] & \text{if} ~\{\tilde{c}(x)=\hat{0}\}\lor\{\exists ~n(x)\in N(x)\cap X_j: \tilde{c}(x)=0\land  \tilde{c}(n(x))=1\}\\ (w'+w(x),\alpha'+\alpha(x)), & \text{if} ~\{\tilde{c}(x)=1\}\land \{w'+w(x)\leq s~\forall(w',\alpha')\in D[j,c'_{0\rightarrow\hat{0}}]\} \\\emptyset & \text{otherwise}
			\end{cases}\]}
		Here, $c'_{0\rightarrow\hat{0}}$ is a coloring that we get by replacing all $0$ colors in $c'$ with $\hat{0}$, which means that for finding $D[i,c]$, we use the coloring $c'_{0\rightarrow\hat{0}}$ instead of $c'$ corresponding to the coloring of $c$. For example, consider $c=(\hat{0},0,1,0)$ and $(\hat{0},\hat{0},1,\hat{0})$ for both of these the corresponding vector will be chosen from $c'_{0\rightarrow\hat{0}}$ is $(\hat{0},\hat{0},1,\hat{0})$.   We need to fill the table $D[i,c]$ for all $c\in \{1,0,\hat{0}\}^{|X_i|}$ and check the validity of each vertex in $X_i$ for the color $\tilde{c}(x)\in \{0,\hat{0}\}$. For the color $\tilde{c}(x)=1$, $D[j,c'_{0\rightarrow\hat{0}}]$ contains at most $\min\{s,\alpha(\VV)\}$ pairs. So total time to compute $D[i,c]$ is $\frac{2}{3}\cdot3^{|X_i|}\cdot|X_i|+\frac{1}{3}\cdot3^{|X_i|}\cdot \min\{s,\alpha(\VV)\}=\OO(3^{|X_i|}\cdot \min\{s,\alpha(\VV)\})$. Ultimately, removing all undominated weight-profit pairs also takes time $\OO(3^{|X_i|}\cdot \min\{s,\alpha(\VV)\})$. So total runtime for introduce node $\OO(3^{|X_i|}\cdot \min\{s,\alpha(\VV)\})$.\\
		
		\item[$\bullet$] If $i$ is \textit{join node} with two children $j$ and $k$:\\ 
		The corresponding bags for $i,j,k$ are $X_i,X_j,X_k$ respectively. By definition $X_i=X_j=X_k=(v_{i_1}, v_{i_2},\ldots, v_{i_{n_i}})$. Let $c\in \{1,0,\hat{0}\}^{|X_i|}$ a coloring for $X_i$. Similarly, $c',c''\in \{1,0,\hat{0}\}^{|X_i|}$ are represent coloring for $X_j$ and $X_k$ respectively. Now we define \textit{consistency} of the colorings $c'$ and $c''$ with the color $c$. To do so, consider $c_t, c_t'$, and $c_t''$ are single components (or "a color") of the vectors $c, c', c''$ respectively. If $c'$ and $c''$ are \textit{consistent} with the color $c$ then the following conditions hold:
		\begin{itemize}
			\item[1.] $c_t'$ = $c_t'' = c_t$ iff $c_t\in \{1,\hat{0}\}$, and
			\item[2.] $(c_t',c_t'')\in\{(0,\hat{0}),(\hat{0},0)\}$ iff $c_t=0$
		\end{itemize}
		Now, we update $D[i,c]$ for all coloring $c$ and for all undominated weight-profit pairs $(w_{c'},\alpha_{c'})\in D[j,c']$ and $(w_{c''},\alpha_{c''})\in D[k,c'']$ as follows:
		{\scriptsize\[ D[i,c] = \{ (w_{c'}+w_{c''}-\mathlarger{\sum}_{\substack{c_l=1\\l\in \{i_1,\ldots, i_{n_i}\} }} w(v_l),\alpha_{c'}+\alpha_{c''}-\mathlarger{\sum}_{\substack{c_l=1\\l\in \{i_1,\ldots, i_{n_i}\} }}\alpha(v_l)): c' \text{and} c'' \text{are consistent with}~ c   \}\]}\\
		Here, $w(v_l)$ and $\alpha(v_l)$ represents the weight and profit for the vertex $v_l\in X_i$ respectively. 
		
		The complexity of finding consistency of $(c',c'')$ with $c$ for each coloring $c$ we need $\OO(4^{|X_i|})$ operations as from the above conditions 1 and 2, we process four 3-tuples $(c_t,c_t',c_t'')$ with the values $(1,1,1)$,$(\hat{0},\hat{0},\hat{0})$,$(0,0,\hat{0})$,$(0,\hat{0},0)$~\cite{alber2002fixed}. Now, each table $D[j,c']$ and $D[k,c'']$ for the two children of node $i$ contains at most $\min\{s,\alpha(\VV)\}$ pairs.
		The complexity of finding $D[i,c]$ based on consistency of $(c',c'')$ with $c$ is $\OO(4^{|X_i|}\cdot \min\{s^2,{\alpha(\VV)}^2\})$.  Then we remove all such pairs from $D[i,c]$ such that $(w'+w''-\sum_{\substack{c_l=1\\l\in \{i_1,\ldots, i_{n_i}\} }} w(v_l))\leq s~ \forall ~c$. Execution of this condition takes $\OO(3^{|X_i|}\cdot \min\{s,\alpha(\VV)\})$. Then, we remove all dominated and duplicated pairs for $D[i,c]~ \ forall~c$. This removal also takes $\OO(3^{|X_i|}\cdot \min\{s,\alpha(\VV)\})$ time. Join node operations takes total time $\OO(4^{|X_i|}\cdot \min\{s^2,{\alpha(\VV)}^2\})$.
	\end{itemize}
	\parindent=.6cm
	Once we obtain $D[r,c]$ for the root node $r$ of the tree decomposition $T^\chi_r$, and the algorithm outputs the weight-profit pairs such that each of them is a \dos. According to Lemma~\ref{lm:ntd} $T^\chi_r$ has $\OO(n)$ nodes. So, the runtime of the algorithm for $\OO(n)$ nodes is $\OO(4^{|X_i|}\cdot n^{\OO(1)} \min\{s^2,{\alpha(\VV)}^2\})$. We compute our final result  $(\hat{w},\hat{\alpha})$, where $\mathlarger{\max}_{\hat{\alpha}\geq \alpha} ~\{\mathlarger{\bigcup}_{\substack{ (w,\alpha)\in D[r,c]\forall c\\c_t\in\{0,1\}}} (w,\alpha) \}$. This additional step takes $\OO(3^{|X_i|}\cdot \min\{s,\alpha(\VV)\})$ time. As  $|X_i|=\OO(tw)$, thus the total runtime of this algorithm is $\OO(4^{tw}\cdot n^{\OO(1)} \min\{s^2,{\alpha(\VV)}^2\})$.
\end{proof}
\begin{theorem}\label{thm:kdoskp-minkdos-treewidth}
	\kdosknapsack and \minimaldosknapsack can also be solved running time $\OO(4^{tw}\cdot n^{\OO(1)} \min\{s^2,{\alpha(\VV)}^2\})$ by an algorithm parameterized by \tw $tw$, where $s$ is the size of the knapsack and $\alpha(\VV)=\sum_{v\in\VV}\alpha(v)$.
\end{theorem}
We need some modification of the previous algorithm to prove this theorem. 

\begin{corollary}
	\dosknapsack,~\kdosknapsack and \minimaldosknapsack are \FPT parameterized by the \tw.
\end{corollary}

\section{Conclusion}\label{sec:conclusion}

In this paper, we proposed the problem \dosknapsack, which is a variation of the \kp with the association of \dos. We proved that \dosknapsack is strongly \NPC for general graphs and even for the special graph class of bipartite graphs, which are very popular due to their simple structure and properties. We showed that \dosknapsack is weakly \NPC for the star graphs and provided a Pseudo-polynomial time algorithm for trees. Also, we proved some parameterized results on \dosknapsack that it is \WTh-complete parameterized by the solution size, and it is very unlikely that \dos has a sub-exponential time algorithm that has running time $f(k)n^{o(k)}$ parameterized by the solution size $k$. For two parameters \tw and \vcknapsack, we showed that \dosknapsack is Fixed Parameter Tractable. We obtained some similar results for other variants, such as \kdosknapsack and \minimaldosknapsack. We have already discussed the application of this problem. Now, we conclude with the statement that, although this problem is generally intractable, it has enormous practical importance in various real-world applications across different fields. 

\bibliography{references}

@inproceedings{DBLP:conf/stoc/GareyJS74,
	author       = {M. R. Garey and
	David S. Johnson and
	Larry J. Stockmeyer},
	editor       = {Robert L. Constable and
	Robert W. Ritchie and
	Jack W. Carlyle and
	Michael A. Harrison},
	title        = {Some Simplified NP-Complete Problems},
	booktitle    = {Proc. 6th Annual {ACM} Symposium on Theory of Computing,
	April 30 - May 2, 1974, Seattle, Washington, {USA}},
	pages        = {47--63},
	publisher    = {{ACM}},
	year         = {1974},
	url          = {https://doi.org/10.1145/800119.803884},
	doi          = {10.1145/800119.803884},
	bibsource    = {dblp computer science bibliography, https://dblp.org}
}

@book{cygan2015parameterized,
	title={Parameterized algorithms},
	author={Cygan, Marek and Fomin, Fedor V and Kowalik, {\L}ukasz and Lokshtanov, Daniel and Marx, D{\'a}niel and Pilipczuk, Marcin and Pilipczuk, Micha{\l} and Saurabh, Saket},
	volume={5},
	number={4},
	year={2015},
	publisher={Springer}
}

@article{alber2002fixed,
	title={Fixed parameter algorithms for dominating set and related problems on planar graphs},
	author={Alber and Bodlaender and Fernau and Kloks and Niedermeier},
	journal={Algorithmica},
	volume={33},
	pages={461--493},
	year={2002},
	publisher={Springer}
}

@article{cygan2011dominating,
	title={Dominating set is fixed parameter tractable in claw-free graphs},
	author={Cygan, Marek and Philip, Geevarghese and Pilipczuk, Marcin and Pilipczuk, Micha{\l} and Wojtaszczyk, Jakub Onufry},
	journal={Theoretical Computer Science},
	volume={412},
	number={50},
	pages={6982--7000},
	year={2011},
	publisher={Elsevier}
}

@inproceedings{lokshtanov2011known,
	title={Known algorithms on graphs of bounded treewidth are probably optimal},
	author={Lokshtanov, Daniel and Marx, D{\'a}niel and Saurabh, Saket},
	booktitle={Proceedings of the twenty-second annual ACM-SIAM symposium on Discrete Algorithms},
	pages={777--789},
	year={2011},
	organization={SIAM}
}

@article{eisenbrand2004complexity,
	title={On the complexity of fixed parameter clique and dominating set},
	author={Eisenbrand, Friedrich and Grandoni, Fabrizio},
	journal={Theoretical Computer Science},
	volume={326},
	number={1-3},
	pages={57--67},
	year={2004},
	publisher={Elsevier}
}

@inproceedings{fomin2012linear,
	title={Linear kernels for (connected) dominating set on H-minor-free graphs},
	author={Fomin, Fedor V and Lokshtanov, Daniel and Saurabh, Saket and Thilikos, Dimitrios M},
	booktitle={Proceedings of the twenty-third annual ACM-SIAM symposium on Discrete Algorithms},
	pages={82--93},
	year={2012},
	organization={SIAM}
}

@article{chang1998efficient,
	title={Efficient algorithms for the domination problems on interval and circular-arc graphs},
	author={Chang, Maw-Shang},
	journal={SIAM Journal on computing},
	volume={27},
	number={6},
	pages={1671--1694},
	year={1998},
	publisher={SIAM}
}

@article{farber1985domination,
	title={Domination in permutation graphs},
	author={Farber, Martin and Keil, J Mark},
	journal={Journal of algorithms},
	volume={6},
	number={3},
	pages={309--321},
	year={1985},
	publisher={Elsevier}
}

@article{farber1984domination,
	title={Domination, independent domination, and duality in strongly chordal graphs},
	author={Farber, Martin},
	journal={Discrete Applied Mathematics},
	volume={7},
	number={2},
	pages={115--130},
	year={1984},
	publisher={Elsevier}
}

@article{cheston1990computational,
	title={On the computational complexity of upper fractional domination},
	author={Cheston, Grant A and Fricke, Gerd and Hedetniemi, Stephen T and Jacobs, D Pokrass},
	journal={Discrete Applied Mathematics},
	volume={27},
	number={3},
	pages={195--207},
	year={1990},
	publisher={Elsevier}
}

@article{cockayne1981contributions,
	title={Contributions to the theory of domination, independence and irredundance in graphs},
	author={Cockayne, Ernest J and Favaron, Odile and Payan, C and Thomason, Andrew G},
	journal={Discrete mathematics},
	volume={33},
	number={3},
	pages={249--258},
	year={1981},
	publisher={Elsevier}
}

@article{jacobson1990chordal,
	title={Chordal graphs and upper irredundance, upper domination and independence},
	author={Jacobson, Michael S and Peters, Ken},
	journal={Discrete Mathematics},
	volume={86},
	number={1-3},
	pages={59--69},
	year={1990},
	publisher={Elsevier}
}

@article{lin2016effective,
	title={An effective hybrid memetic algorithm for the minimum weight dominating set problem},
	author={Lin, Geng and Zhu, Wenxing and Ali, Montaz M},
	journal={IEEE Transactions on Evolutionary Computation},
	volume={20},
	number={6},
	pages={892--907},
	year={2016},
	publisher={IEEE}
}

@article{maw1997weighted,
	title={Weighted domination of cocomparability graphs},
	author={Maw-Shang, Chang},
	journal={Discrete applied mathematics},
	volume={80},
	number={2-3},
	pages={135--148},
	year={1997},
	publisher={Elsevier}
}

@article{natarajan1978optimum,
	title={Optimum domination in weighted trees},
	author={Natarajan, KS and White, Lee J},
	journal={Information processing letters},
	volume={7},
	number={6},
	pages={261--265},
	year={1978},
	publisher={Elsevier}
}

@article{wang2012ptas,
	title={PTAS for the minimum weighted dominating set in growth bounded graphs},
	author={Wang, Zhong and Wang, Wei and Kim, Joon-Mo and Thuraisingham, Bhavani and Wu, Weili},
	journal={Journal of Global Optimization},
	volume={54},
	number={3},
	pages={641--648},
	year={2012},
	publisher={Springer}
}

@article{potluri2013hybrid,
	title={Hybrid metaheuristic algorithms for minimum weight dominating set},
	author={Potluri, Anupama and Singh, Alok},
	journal={Applied Soft Computing},
	volume={13},
	number={1},
	pages={76--88},
	year={2013},
	publisher={Elsevier}
}

@article{alber2005refined,
	title={A refined search tree technique for dominating set on planar graphs},
	author={Alber, Jochen and Fan, Hongbing and Fellows, Michael R and Fernau, Henning and Niedermeier, Rolf and Rosamond, Fran and Stege, Ulrike},
	journal={Journal of Computer and System Sciences},
	volume={71},
	number={4},
	pages={385--405},
	year={2005},
	publisher={Elsevier}
}

@book{downey2013fundamentals,
	title={Fundamentals of parameterized complexity},
	author={Downey, Rodney G and Fellows, Michael R and others},
	volume={4},
	year={2013},
	publisher={Springer}
}

@inproceedings{iwata2012faster,
	title={A faster algorithm for dominating set analyzed by the potential method},
	author={Iwata, Yoichi},
	booktitle={Parameterized and Exact Computation: 6th International Symposium, IPEC 2011, Saarbr{\"u}cken, Germany, September 6-8, 2011. Revised Selected Papers 6},
	pages={41--54},
	year={2012},
	organization={Springer}
}

@article{baker1994approximation,
	title={Approximation algorithms for NP-complete problems on planar graphs},
	author={Baker, Brenda S},
	journal={Journal of the ACM (JACM)},
	volume={41},
	number={1},
	pages={153--180},
	year={1994},
	publisher={ACM New York, NY, USA}
}

@article{van2011exact,
	title={Exact algorithms for dominating set},
	author={Van Rooij, Johan MM and Bodlaender, Hans L},
	journal={Discrete Applied Mathematics},
	volume={159},
	number={17},
	pages={2147--2164},
	year={2011},
	publisher={Elsevier}
}

@inproceedings{markarian2012degree,
	title={A degree-based heuristic for strongly connected dominating-absorbent sets in wireless ad-hoc networks},
	author={Markarian, Christine and Abu-Khzam, Faisal N},
	booktitle={2012 International Conference on Innovations in Information Technology (IIT)},
	pages={200--204},
	year={2012},
	organization={IEEE}
}

@article{weihe1998covering,
	title={Covering trains by stations or the power of data reduction},
	author={Weihe, Karsten},
	journal={Proceedings of Algorithms and Experiments, ALEX},
	pages={1--8},
	year={1998}
}

@book{kloks1994treewidth,
	title={Treewidth: computations and approximations},
	author={Kloks, Ton},
	year={1994},
	publisher={Springer}
}

@book{niedermeier2006invitation,
	title={Invitation to fixed-parameter algorithms},
	author={Niedermeier, Rolf},
	volume={31},
	year={2006},
	publisher={OUP Oxford}
}

@book{martello1990knapsack,
	title={Knapsack problems: algorithms and computer implementations},
	author={Martello, Silvano and Toth, Paolo},
	year={1990},
	publisher={John Wiley \& Sons, Inc.}
}

@book{kellerer2004multidimensional,
	title={Multidimensional knapsack problems},
	author={Kellerer, Hans and Pferschy, Ulrich and Pisinger, David and Kellerer, Hans and Pferschy, Ulrich and Pisinger, David},
	year={2004},
	publisher={Springer}
}

@article{cacchiani2022knapsack,
	author       = {Valentina Cacchiani and
	Manuel Iori and
	Alberto Locatelli and
	Silvano Martello},
	title        = {Knapsack problems - An overview of recent advances. Part {II:} Multiple,
	multidimensional, and quadratic knapsack problems},
	journal      = {Comput. Oper. Res.},
	volume       = {143},
	pages        = {105693},
	year         = {2022},
	url          = {https://doi.org/10.1016/j.cor.2021.105693},
	doi          = {10.1016/j.cor.2021.105693},
	bibsource    = {dblp computer science bibliography, https://dblp.org}
}

@article{CacchianiILM22,
	author       = {Valentina Cacchiani and
	Manuel Iori and
	Alberto Locatelli and
	Silvano Martello},
	title        = {Knapsack problems - An overview of recent advances. Part {I:} Single
	knapsack problems},
	journal      = {Comput. Oper. Res.},
	volume       = {143},
	pages        = {105692},
	year         = {2022},
	url          = {https://doi.org/10.1016/j.cor.2021.105692},
	doi          = {10.1016/j.cor.2021.105692},
	bibsource    = {dblp computer science bibliography, https://dblp.org}
}

@article{PferschyS09,
	author       = {Ulrich Pferschy and
	Joachim Schauer},
	title        = {The Knapsack Problem with Conflict Graphs},
	journal      = {J. Graph Algorithms Appl.},
	volume       = {13},
	number       = {2},
	pages        = {233--249},
	year         = {2009},
	url          = {https://doi.org/10.7155/jgaa.00186},
	doi          = {10.7155/jgaa.00186},
	bibsource    = {dblp computer science bibliography, https://dblp.org}
}

@inproceedings{dey2024knapsack,
	title={Knapsack: Connectedness, path, and shortest-path},
	author={Dey, Palash and Kolay, Sudeshna and Singh, Sipra},
	booktitle={Latin American Symposium on Theoretical Informatics},
	pages={162--176},
	year={2024},
	organization={Springer}
}

@article{dey2024knapsackwith,
	title={Knapsack with Vertex Cover, Set Cover, and Hitting Set},
	author={Dey, Palash and Hota, Ashlesha and Kolay, Sudeshna and Singh, Sipra},
	journal={arXiv preprint arXiv:2406.01057},
	year={2024}
}

@inproceedings{hermelin2011domination,
	title={Domination when the stars are out},
	author={Hermelin, Danny and Mnich, Matthias and Van Leeuwen, Erik Jan and Woeginger, Gerhard J},
	booktitle={Automata, Languages and Programming: 38th International Colloquium, ICALP 2011, Zurich, Switzerland, July 4-8, 2011, Proceedings, Part I 38},
	pages={462--473},
	year={2011},
	organization={Springer}
}

@article{bellman1954some,
	title={Some applications of the theory of dynamic programming—a review},
	author={Bellman, Richard},
	journal={Journal of the Operations Research Society of America},
	volume={2},
	number={3},
	pages={275--288},
	year={1954},
	publisher={INFORMS}
}

@incollection{karp2009reducibility,
	title={Reducibility among combinatorial problems},
	author={Karp, Richard M},
	booktitle={50 Years of Integer Programming 1958-2008: from the Early Years to the State-of-the-Art},
	pages={219--241},
	year={2009},
	publisher={Springer}
}

@misc{gambosi1999complexity,
	title={Complexity and approximation},
	author={Gambosi, G Ausiello P Crescenzi G and Marchetti-Spaccamela, V Kann A and Protasi, M},
	year={1999},
	publisher={Springer Berlin}
}
\end{document}